\documentclass[11pt,reqno]{article}
\usepackage[utf8]{inputenc}
\usepackage{sty}
\usepackage{fullpage}
\usepackage[margin=1in]{geometry}
\usepackage{hyperref}

\usepackage{authblk}
\author[1]{David Gamarnik\thanks{Email: \textit{gamarnik@mit.edu}. Supported by NSF grant DMS-2015517.}}
\author[2]{Aukosh Jagannath\thanks{Email: \textit{a.jagannath@uwaterloo.ca}.
Supported by NSERC Discovery grant RGPIN-2020-04597 and DGECR-2020-00199.}}
\author[3]{Alexander S.\ Wein\thanks{Email: \textit{awein@cims.nyu.edu}. Supported by NSF awards CCF-2007443 and CCF-2106444. Part of this work was done while with the Courant Institute of Mathematical Sciences at New York University, partially supported by NSF grant DMS-1712730 and by the Simons Collaboration on Algorithms and Geometry. Part of this work was done at the Simons Institute for the Theory of Computing at UC Berkeley, supported by a Simons-Berkeley Research Fellowship.}}

\affil[1]{Sloan School of Management and Operations Research Center, MIT}
\affil[2]{Department of Statistics and Actuarial Science and Department of Applied Mathematics, University of Waterloo}
\affil[3]{Algorithms and Randomness Center, Georgia Tech}

\date{}

\title{Hardness of Random Optimization Problems for Boolean Circuits, Low-Degree Polynomials, and Langevin Dynamics
\footnote{Conference version~\cite{gamarnik2020lowFOCS} appeared in Proceedings of Foundations of Computer Science (FOCS) 2020.}}

\begin{document}

\maketitle

\begin{abstract}
We consider the problem of finding nearly optimal solutions of optimization problems with random objective functions. Such problems arise widely in the theory of random graphs, theoretical computer science, and statistical physics. Two concrete problems we consider are (a) optimizing the Hamiltonian of a spherical or Ising $p$-spin glass model, and (b) finding a large independent set in a sparse \ER graph. The following families of algorithms are considered: 
(a) low-degree polynomials of the input---a general framework that captures many prior algorithms;  
(b) low-depth Boolean circuits;
(c) the Langevin dynamics algorithm, a canonical Monte Carlo analogue of the gradient descent algorithm.
We show that these families of algorithms fail to produce nearly optimal solutions with high probability. 
For the case of Boolean circuits, our results improve the state-of-the-art bounds known in circuit
complexity theory (although we consider the search problem as opposed to the decision problem).

Our proof uses the fact that these models are known to exhibit a variant of the overlap gap property (OGP) of near-optimal solutions. Specifically, for both models, every two solutions whose objectives are above a certain threshold are either close or far from each other. The crux of our proof is that the classes of algorithms we consider exhibit a form of stability (noise-insensitivity): a small perturbation of the input induces a small 
perturbation of the output. We show by an interpolation argument that stable algorithms cannot overcome the OGP barrier.

The stability of Langevin dynamics is an immediate consequence of the well-posedness of stochastic differential equations. The stability of low-degree polynomials and Boolean circuits  is established using tools from Gaussian and Boolean analysis---namely hypercontractivity and total influence, as well as a novel lower bound for random walks  avoiding certain subsets, which we expect to be of independent interest. In the case of Boolean circuits, the result also makes use of Linal--Mansour--Nisan's classical theorem. Our techniques apply more broadly to low influence functions and we expect that they may apply more generally.
\end{abstract}

\clearpage

\maketitle

\section{Introduction}
We study the problem of producing near-optimal solutions of random optimization problems by several classes of algorithms including polynomials of low degree, Boolean circuits of low depth, and Langevin dynamics with limited time horizon. We prove that these algorithms cannot succeed at achieving a certain objective value in two types of optimization problems: (a) optimizing the Hamiltonian of the (spherical or Ising) $p$-spin glass model (that is, finding a near ground state), 
and (b) finding a large independent set in a sparse \ER graph, with high probability (w.h.p.)\ in the realization of the problem.

For the $p$-spin model optimization problem, we rule out polynomials of degree linear in  $n$ (with appropriately small constant in front)
provided the algorithm is assumed to succeed modulo exponentially small in $n$ probability, where $n$ is the problem dimension. 
More generally, we provide a tradeoff between the degree of polynomials that we rule out and the success probability assumed.
Our main result regarding Langevin dynamics applies (for natural reasons) only to the spherical $p$-spin model,
where the search space is continuous. We establish that with high probability, Langevin dynamics fails to find a nearly optimal solution
in time which is linear in the model dimension.

For the independent set problem, we focus our presentation on Boolean circuits. We prove a superpolynomial lower bound on the size of Boolean circuits which find
better than half optimum independent sets in sparse random graphs and have depth $O(\log n/\log\log n)$. 
Prior results of this kind require depth $o(\log n/\log\log n)$~\cite{rossman2010average,li2017ac} 
(although for the decision as opposed to search problem). While we state our results for Boolean circuits, 
our arguments apply more general to algorithms with low total influence
and we expect that they will be useful for other classes of algorithms.

\subsection*{Algorithmic classes: power, limits, and our results}
We now discuss each of the three classes of algorithms considered in this paper, including the motivations for studying these algorithm, prior results, and our contributions.

\subsubsection*{Algorithms based on low-degree polynomials}
Our motivation for studying low-degree polynomials is two-fold. Firstly, from an approximation theoretic perspective, producing near-optimal solutions by a polynomial in the input is natural. Furthermore, the best known polynomial-time algorithms in many problems of interest can be placed within the family of low-degree methods. For example, in the settings we consider here, the best known polynomial-time optimization results can be captured by the approximate message passing (AMP) framework~\cite{montanari-sk,EMS-opt,sellke2021optimizing} 
(for the $p$-spin) and by the class of local algorithms on sparse graphs~\cite{LauerWormald} (for the independent set problem), respectively. Both of these families of algorithms are captured by constant-degree polynomials; see Appendix~\ref{app:low-deg-alg} for more details.
For spherical $p$-spin glass models, earlier work of \cite{subag-full-rsb} introduced an algorithm which performs as well as AMP; we expect this algorithm to also fall into the family of low-degree methods, but verifying this is less clear. 

Secondly, a recent line of work~\cite{p-cal,HS-bayesian,sos-hidden,sam-thesis} on the \emph{sum-of-squares hierarchy} has produced compelling evidence that the power of low-degree polynomials is a good proxy for the intrinsic computational complexity of a broad class of \emph{hypothesis testing} problems. This line of work has focused on high-dimensional testing problems where the goal is to determine whether a given sample (e.g., an $n$-vertex graph) was drawn from the ``null'' distribution $\QQ_n$ (e.g., the \ER model) or the ``planted'' distribution $\PP_n$ (e.g., a random graph with planted structure such as a large clique or a small cut). Through an explicit and relatively straightforward calculation, one can determine whether there exists a (multivariate) polynomial of a given degree $D = D(n)$, $f$, of the entries of the sample that can distinguish $\PP_n$ from $\QQ_n$ (in a particular sense) \cite{HS-bayesian,sos-hidden,sam-thesis}. A conjecture of Hopkins~\cite{sam-thesis} (inspired by~\cite{p-cal,HS-bayesian,sos-hidden}) postulates that for ``natural'' high-dimensional testing problems, if there is a polynomial-time algorithm to distinguish $\PP_n$ from $\QQ_n$ (with error probability $o(1)$) then there is also an $O(\log n)$-degree polynomial that can distinguish $\PP_n$ from $\QQ_n$. One justification for this conjecture is its deep connection with the \emph{sum-of-squares (SoS) hierarchy}---a powerful class of meta-algorithms---and in particular the \emph{pseudo-calibration} approach~\cite{p-cal}, which suggests that low-degree polynomials are as powerful as any SoS algorithm (see~\cite{sos-hidden,sam-thesis,sos-survey} for details). Another justification for the conjecture is that $O(\log n)$-degree polynomials can capture a very broad class of spectral methods (see~\cite[Theorem~4.4]{lowdeg-notes} for specifics), which in turn capture the best known algorithms for many high-dimensional testing problems (e.g., \cite{tensor-pca-sos,fast-sos,sos-hidden}). For many classical statistical tasks---planted clique, sparse PCA, community detection, tensor PCA, etc.---it has indeed been verified that $O(\log n)$-degree polynomials succeed (at testing) in precisely the same parameter regime as the best known polynomial-time algorithms (e.g.,~\cite{HS-bayesian,sos-hidden,sam-thesis,sk-cert,lowdeg-notes,subexp-sparse}). (Oftentimes, the hypothesis testing variants of these types of problems seem to be equally hard as the more standard task of recovering the planted signal.) Lower bounds against low-degree polynomials are one concrete form of evidence that the existing algorithms for these problems cannot be improved (at least without drastically new algorithmic techniques). For more details on the low-degree framework for hypothesis testing, we refer the reader to~\cite{sam-thesis,lowdeg-notes}.

Motivated by this, one goal of this paper is to extend the low-degree framework to the setting of random optimization problems. This includes defining what it means for a low-degree polynomial to succeed at an optimization task and giving techniques by which one can prove lower bounds against all low-degree polynomials. 

In this paper, we show a linear in $n$ lower bound on the degree required for such an algorithm to succeed modulo exponentially small probability in $n$ (the dimension of the problem) and more generally obtain a tradeoff between the degree and the success probability.
A substantial difficulty that we face in the optimization setting as compared to the testing setting is that it does not seem possible to prove lower bounds against low-degree polynomials via an explicit linear-algebraic calculation. 
To overcome this, our proofs take a more indirect route and leverage a certain structural property---the \emph{overlap gap property (OGP)}---of the optimization landscape. The core idea of this paper is that algorithms with low total influence, appropriately defined, exhibit a certain stability property that prevents them from overcoming the OGP barrier. We will discuss the OGP and refutations proved using it in detail momentarily.

For the case of low-degree polynomials, the stability property we require is Theorem~\ref{thm:hyp-stable} which states that if we have two $\rho$-correlated random instances $X$ and $Y$ of a random Gaussian tensor, and $f$ is a {vector-valued} low-degree polynomial defined on such tensors, then the distance, $\|f(X)-f(Y)\|_2$, is unlikely to exceed a certain value which depends continuously on $\rho$. In particular this distance is small when $\rho\approx 1$. This result relies on well-known consequences of hypercontractivity for low-degree polynomials
and basic properties of Hermite polynomials (the orthogonal polynomials of the Gaussian measure).

We emphasize that the class of algorithms based on low-degree polynomials which are ruled out by our lower bounds not only captures existing methods, such as AMP and local algorithms, but contains a strictly larger (in a substantial way) class of algorithms than prior work on random optimization problems. For example, the best known polynomial-time algorithms for optimizing the $p$-spin Hamiltonian are captured by the AMP framework \cite{montanari-sk,EMS-opt}. Roughly speaking, AMP algorithms combine a linear update step (tensor power iteration) with entry-wise non-linear operations. For a fairly general class of $p$-spin optimization problems (including spherical and Ising mixed $p$-spin models), it is now known precisely what objective value can be reached by the best possible AMP algorithm~\cite{EMS-opt}. While this may seem like the end of the story, we point out that for the related \emph{tensor PCA} problem---which is a variant of the $p$-spin model with a planted rank-1 signal---AMP is known to be substantially sub-optimal compared to other polynomial-time algorithms~\cite{RM-tensor}. None of the best known polynomial-time algorithms~\cite{RM-tensor,tensor-pca-sos,fast-sos,kikuchi,hastings-quantum,replicated-gradient} use the tensor power iteration step as in AMP, and there is evidence that this is fundamental~\cite{algorithmic-tensor}; instead, the optimal algorithms include spectral methods derived from different tensor operations such as \emph{tensor unfolding}~\cite{RM-tensor,tensor-pca-sos} (which can be interpreted as a higher-order ``lifting'' of AMP~\cite{kikuchi}). These spectral methods are captured by $O(\log n)$-degree polynomials. With this in mind, we should \emph{a priori} be concerned that AMP might also be sub-optimal for the (non-planted) $p$-spin optimization problem. This highlights the need for lower bounds that rule out not just AMP, but all low-degree polynomial algorithms. While the lower bounds in this paper do not achieve the precise optimal threshold for objective value (see~\cite{huang2021tight} for a recent improvement), they rule out quite a large class of algorithms compared to existing lower bounds for random optimization problems.
We refer the reader to Appendix~\ref{app:low-deg-alg} for a more detailed discussion of how various optimization algorithms can be approximated by low-degree polynomials.

\subsubsection*{Bounded depth Boolean circuits}
Next we discuss algorithms based on Boolean circuits, as well as the ensuing stability property. 
Boolean circuits constitute one of the standard models for understanding algorithmic tractability 
and hardness in combinatorial optimization problems~\cite{arora2009computational,sipser,alon2004probabilistic}. 
One potential route to proving the widely-believed
conjecture $P\ne NP$ would be to show non-existence of polynomial-size Boolean circuits
solving problems in $NP$. A large body of literature in circuit complexity is devoted
to establishing limits on the power of circuits with bounded depth, specifically
the power of constant-depth circuits known as AC[0] circuits and its immediate extension---circuits with nearly logarithmic (in the problem size) depth. Many hardness results have been
obtained by working on models with random inputs. A classical result is one by H\r{a}stad~\cite{hastad1986almost}
who established a $\log n/(\log\log n+O(1))$ lower bound on the depth of any poly-size circuit that computes the $n$-parity function. A lot of focus is on circuit complexity for computing various
natural combinatorial optimization problems.
For example Rossman~\cite{rossman2010average} has shown
that poly-size circuits detecting the presence of a $k(n)$-size clique in a graph must have 
depth at least $\Omega(\log n/(k^2(n)\log\log n))$, where $n$ is the number of graph nodes and $k(n)$
is any function growing in $n$.
This was obtained by working with the sparse random \ER  graph model with 
the edge probability tuned so that the largest ``naturally occurring'' clique is of size smaller than $k(n)$,
with high probability as $n$ increases. In fact, for combinatorial optimization problems this 
represents the state-of-the-art lower bound on circuit depth for polynomial-size circuits, though many extensions exist for subgraph homomorphism existence 
problems~\cite{rossman2018lower,li2017ac}.

In this paper, we establish a lower bound of $\log n/((1+\epsilon)\log\log n)$ on the depth of any polynomial-size Boolean circuit which solves a similar combinatorial optimization problem: largest independent set of a graph. 
Specifically, considering independent sets in sparse \ER graphs $\G(n,d/n)$ (graphs on $n$ nodes
with each edge present with probability $d/n$ independently), we show that any poly-size circuit with the aforementioned
bound on depth fails to find independent sets which are larger than half optimum in this graph, 
in the doubly asymptotic regime
$n\to\infty,d\to\infty$. Half-optimality is largely believed to be the edge of algorithmic tractability for this 
problem~\cite{gamarnik2014limits,rahman2017local,coja2011independent,wein2020optimal}.
Notably, unlike the aforementioned result for cliques, our lower bound does 
not depend on the objective value, such as $k(n)$, and thus our bound is 
$O(\log n/\log\log n)$ as opposed to $o(\log/\log\log n)$. Our result though, strictly speaking, does not improve upon
\cite{rossman2010average}, since our circuits are supposed to produce the \emph{entire solution}, as opposed to 
just solving
the associated YES/NO decision problem. Nevertheless,  our result presents 
the strongest known circuit depth lower bounds for combinatorial optimization problems, to the best of our knowledge.
An arXiv note accompanying this paper~\cite{gamarnik2021circuit} establishes a similar result on circuit complexity 
for the $p$-spin model exactly in the regime
where low-degree polynomials are ruled out per the discussion above. 

As for the case of the $p$-spin optimization problem, we employ a version of the OGP exhibited by the space of 
large independent sets in the graph $\G(n,d/n)$. In particular, 
we use a multi-overlap version of the OGP which rules out the presence of a certain $m$-tuple
of larger-than-half-optimum independent sets in a family of correlated $\G(n,d/n)$ graphs with some prescribed intersection 
properties. A precise formulation is given in Theorem~\ref{thm:multi-OGP}. 
This version of the the multi-OGP was established recently by the third author in~\cite{wein2020optimal} and used to rule out methods based on low-degree polynomials for the same problem building on the conference version of the 
current paper~\cite{gamarnik2020lowFOCS}). Here we take the approach a step further and use this multi-OGP structure to rule out poly-size Boolean circuits with the stated bound on depth.  
The  proof idea is again to pit OGP against stable algorithms. 
Poly-size circuits with the stated bound on depth are known to have bounded total influence by the celebrated Linial--Mansour--Nisan (LMN) Theorem~\cite{LMN}. 
(The version we use is found in~\cite[Thm 4.30]{o-book}.) The tightest known bound of this kind
was obtained by Tal in~\cite{tal2017tight}.
This is used to establish
the stability of circuits and ultimately its failure to overcome the OGP barrier. 

Two complications do arise however in our analysis. First, in light of the fact that our underlying model
is discrete, the approach based on correlated Gaussians used for $p$-spin models does not appear to be applicable.
Instead we consider a discrete interpolation scheme where a family of correlated random graphs are obtained
by resampling one edge at a time. For this sequence, in order to ensure stability, we need to guarantee 
that \emph{every single} resampling step results in a small change of the circuit's output. While the total influence bounds provide an upper bound on the total \emph{expected} number of resampling steps violating this condition, 
they do not immediately provide a guarantee of the \emph{for all} type. For this purpose we derive a technical lemma (Proposition~\ref{prop:no-jump-events}) with the following informal statement. 
Suppose any collection of  pairs $(u,v)$ for $u,v\in\{\pm 1\}^m$ with Hamming distance $1$
are marked. Suppose the fraction of marked edges is $x$ (so that the total number is $x2^{m-1}m$). 
Consider a random walk on $\{\pm 1\}^m$ run for {$m$ } steps starting from a uniform random point. 
We prove that with probability at least {$2^{-xm}$}, this random walk will \emph{never} encounter a marked edge.  (For our purposes, marked edges will correspond to resampling steps which result in a large change of the output, but the lemma applies for \emph{any} collection of marked edges.)  This lemma provides the tool for 
relating the likelihood of
\emph{for all} type to the total expected likelihood. We believe that this lemma is of independent interest.
An interesting open question
is whether an analogue of this statement holds for continuous-time continuous-space processes, such as Brownian motion
on a sphere.

The second complication arises because we are dealing with a biased Bernoulli distribution (with parameter
$d/n$), whereas the LMN theorem \emph{a priori} applies to the unbiased (uniform) distribution. While there are extensions
of the LMN theorem to biased cases (see for example Lemma~9 in~\cite{furst1991improved}), these extensions  
appear to be too weak for our purposes. Instead we resort to a trick which relates a biased Bernoulli distribution 
to an unbiased Bernoulli
distribution  as follows. Consider a uniformly generated random $0/1$ string
of $\log_2 n$ bits and the function of it which outputs $1$ if the value encoded by this string (using the binary
representation) is at most $d-1$, and $0$ otherwise. Note that the likelihood of obtaining value $1$ is exactly $d/n$
and the likelihood of obtaining value $0$ is $1-d/n$.
Thus we can simulate an adjacency matrix of a graph by creating these $(\log_2 n){n\choose 2}$ bits and taking a uniform
distribution on them. It is straightforward to check that the translation from bits to the adjacency matrix can be 
coded by some depth-3 poly-size circuit $\Phi$. 
Then instead of analyzing the performance of a circuit $C$ acting on the adjacency 
matrix, we analyze the performance of this circuit composed with $\Phi$, i.e., $C\circ \Phi$. We can then use the canonical version of the LMN theorem applied to the uniform distribution to obtain our required stability result.

Our result raises the question of whether a similar circuit lower bound can be established for dense \ER graphs
$\G(n,p)$ where each edge is present with a constant probability $p$ (say $1/2$) independently 
for all edges. Unfortunately, while the result  very likely holds, our proof technique comes short.
Specifically, while the multi-OGP described for the sequences of sparse \ER graphs extends to the dense
case as well (this can be easily verified, though not formally written down in any published papers), 
the high-probability likelihood of the multi-OGP is only of the form $1-\exp(-\Theta(\log^2 n))$, 
as opposed to $1-\exp(-\Theta(n))$ in the sparse case. This turns out to be too weak to rule out Boolean
circuits of any depth. The exponential tail of the probability above for the case of sparse graphs turns
out to be crucial for our proof technique to apply. We leave it as an interesting open question
to rule out Boolean circuits for finding independent sets (or equivalently cliques) in dense \ER
graphs  even at constant depth, whether using the OGP-based or any other approach.

\subsubsection*{Langevin dynamics}
Finally, we briefly discuss our results for Langevin dynamics. As mentioned earlier, this algorithm,
which is a continuous-time analogue of a Markov Chain Monte Carlo method, is applicable only to the spherical $p$-spin model.
We  use similar techniques based on the OGP to give lower bounds against Langevin dynamics, when it is run
for linear time (with small enough constant). Specifically, we controls the divergence of two paths of the dynamics
on two correlated instances via Gronwall's inequality to show that the algorithm is stable: a small change in the $p$-spin Gaussian coupling
results in a small change in the algorithm output when it is run for linear time. 
We use this stability result to show that Langevin dynamics on this time horizon cannot overcome the OGP barrier. We conjecture
that Langevin dynamics needs to be run for exponentially long time to achieve near optimality, but our current
technique is not capable of showing this. (We note here that our argument also applies to gradient flow which can be viewed as Langevin dynamics at ``zero temperature''.)

\subsection*{OGP as a Barrier to Algorithms}
We now provide some background on OGP and its variants, as well as its uses in some prior works. 
A recent survey on the topic by the first author is~\cite{gamarnik2021overlap}. 
While the OGP has already been used to rule out various classes of other algorithms in previous works 
(see below), its usage in our current setting, specifically the need to establish stability of algorithms,
 presents some substantial technical difficulties which we need to overcome.  Roughly speaking, the OGP states that for every pair of  solutions $x_1$ and $x_2$ achieving some proximity to optimality, the absolute value of their normalized  overlap (normalized inner product) measured with respect to the ambient Hilbert space must lie in a disjoint union of intervals $[0,\nu_1]\cup [\nu_2,1]$. 

An extended version of this property, which we call the ensemble-OGP corresponds to the case of families 
of instances in the sense that  if one considers a natural interpolation between two independent instances 
of the problem, then the ensemble-OGP holds if 
for every two members of the interpolated family and every pair of solutions $x_1,x_2$ which 
are near optimizers for these two members, respectively, it is still the case that the overlap of $x_1$ and $x_2$ 
belongs to $[0,\nu_1]\cup [\nu_2,1]$. The main idea of the proof from the ensemble-OGP is based on the 
following contradiction argument. 
If the result of the algorithm is known to be stable then, denoting by $x(t)$ the result of the algorithm 
corresponding to the interpolation step $t$, it should be the case that the overlap between $x(0)$ and $x(t)$ 
changes ``continuously''. At the same time we show separately that the starting solution $x(0)$ and terminal 
solution $x(1)$ have an overlap at most $\nu_1$,  and thus at some point the overlap between $x(0)$ and $x(t)$ 
belongs to $(\nu_1,\nu_2)$, which is a contradiction. This type of contradiction argument can be used
to prove that algorithms cannot achieve a certain multiplicative constant guarantee from optimality.
Unfortunately, it is usually not enough to prove that such algorithms fail beyond the known algorithmic bounds.
For this purpose a more powerful extension of the OGP is considered (both for the single and ensemble versions),
namely the multi-overlap version which considers $m$-tuples of near optimal solutions and rules out
particular patterns of intersections between them.

The OGP emerged for the first time in the context of spin glasses and random constraint satisfaction problems. It was first established implicitly in~\cite{achlioptas2006solution} and \cite{mezard2005clustering}. These papers established that the set of satisfying assignments of a random $k$-SAT formula partitions into clusters above a certain clause-to-variable density. This was postulated as evidence for algorithmic hardness of finding satisfying assignments at such densities. Implicitly, the proof of the cluster existence revealed that the overlaps of satisfying assignments exhibit the OGP, and the clustering property was obtained as an implication. 
It is worth noting that while the OGP implies the existence of clusters, the converse is not necessarily the case, and we refer to~\cite{gamarnik2021overlap} for a discussion of examples based on the binary perceptron model, as well as the distinctions between the weak and strong forms of clustering.

A provable algorithmic implication of the OGP was established for the first time in~\cite{gamarnik2014limits}, where 
OGP was proven to be a barrier for local algorithms---defined as the 
so-called \emph{factors of i.i.d.\ (FIID)}---designed to find large independent sets in sparse \ER graphs
$\G(n,d/n)$. This is the same class of graphs we consider in this paper.
The OGP was used to show that, asymptotically, these algorithms cannot find independent sets larger than a multiplicative factor $1/2+1/(2\sqrt{2}) \approx 0.85$ of the optimal one.   
This was improved to match the best known  algorithmic threshold $1/2$
by Rahman and Vir\'ag in~\cite{rahman2017local} by introducing for the first time the multi-overlap version of OGP. 
Since then, the multi-overlap versions of the OGP have turned out to be a very effective method of ruling out algorithms 
all the way to the known algorithmic thresholds, including~\cite{gamarnik2017performance, chen2019suboptimality, wein2020optimal,bresler2021algorithmic,huang2021tight}. In an earlier (conference) 
version of this paper, the result
in~\cite{gamarnik2014limits} for the $1/2+1/(2\sqrt{2})$ approximation factor for independent sets 
was generalized to the class of low-degree polynomial based algorithms.
Then in the work of the third author~\cite{wein2020optimal} this was further improved to the known algorithmic threshold $1/2$ based on a novel (asymmetric) version of the multi-OGP. It is this version that we employ
in our result for Boolean circuits. Unfortunately a (symmetric) version of the  multi-OGP considered
earlier in~\cite{rahman2017local} does not appear to be useful for either low-degree polynomial or circuits.
We discuss in Appendix~\ref{app:low-deg-alg} how local algorithms can be captured by constant-degree 
polynomials. 

Several other papers used OGP to rule out various classes of algorithms, including local algorithms 
for finding large cuts in random hypergraphs~\cite{chen2019suboptimality}, random walk--based 
algorithms (WALKSAT)~\cite{coja2017walksat}, and Lipschitz iteration schemes (such as AMP) 
for optimizing the Hamiltonian of the 
Ising $p$-spin model~\cite{gamarnik2019overlap}. Also several recent works established the presence of the
(multi)-OGP
all the way to the known algorithmic thresholds, including~\cite{bresler2021algorithmic} for low-degree polynomials 
for the random $k$-SAT model, and~\cite{huang2021tight} for general classes of algorithms exhibiting
Lipschitz type form of stability for the $p$-spin models. The latter work uses a particularly ingenious form of the multi-OGP
inspired by the ultrametricity properties exhibited by the $p$-spin models~\cite{Mez84,panchenko2013sherrington}.

\subsubsection*{Notation}

We use $\| \cdot \|_2$ and $\langle \cdot,\cdot \rangle$ to denote the standard $\ell^2$ norm and inner product of vectors. We also use the same notation to denote the Frobenius norm and inner product of tensors. We use the term \emph{polynomial} both to refer to (multivariate) polynomials $\RR^m \to \RR$ in the usual sense, and to refer to vector-valued polynomials $\RR^m \to \RR^n$ defined as in~\eqref{eq:vec-val-poly}. We abuse notation and use the term \emph{degree-$D$ polynomial} to mean a polynomial of degree \emph{at most} $D$. A \emph{random polynomial} has possibly-random coefficients, as defined in Section~\ref{sec:poly-alg}. We use $A^c$ to denote the complement of an event $A$. Unless stated otherwise, asymptotic notation such as $o(1)$ or $\Omega(n)$ refers to the limit $n \to \infty$ with all other parameters held fixed. In other words, this notation may hide constant factors depending on other parameters such as the degree $d$ of the graph in the independent set problem.
We use $\sim$ to denote equality in distribution. Thus for two random variables $X$ and $Y$, $X\sim Y$ means
$X$ and $Y$ have an identical distribution.

\section{Main Results}
\label{sec:main-results}

\subsection{Optimizing the $p$-Spin Glass Hamiltonian}
\label{sec:p-spin}

The first class of problems we consider here is optimization of the (pure) $p$-spin glass Hamiltonian, defined as follows. Fix an integer $p \geq 2$ and let $Y \in (\RR^n)^{\otimes p}$ be a $p$-tensor with real coefficients. For $x \in \RR^n$, consider the objective function
\begin{equation}\label{eq:p-spin-def}
H_n(x;Y) =  \frac{1}{n^{(p+1)/2}} \g{Y,x^\tp}.
\end{equation}
Note that all homogeneous polynomials of degree $p$ (in the variables $x$) can be written in this form for some $Y$. We focus on the case of a random coefficient tensor $Y$. In this setting, the function $H_n$ is 
sometimes called the  Hamiltonian for a $p$-spin glass model in the statistical physics literature. More precisely, for various choices of a (compact) domain $\cX_n \subset \R^n$, we are interested in approximately solving the optimization problem
\begin{equation}\label{eq:max-H}
\max_{x \in \cX_{n}} H_n(x;Y)
\end{equation}
given a random realization of the coefficient tensor $Y$ with i.i.d.\ $\mathcal{N}(0,1)$ entries. Here and in the following we let $\PP_Y$ denote the law of $Y$. (When it is clear from context we omit the subscript $Y$.)

We begin first with a simple norm constraint, namely, we will take as domain 
$\cS_n =\{x\in\R^n: \norm{x}_2 =\sqrt n\}$, the sphere in $\R^n$ of radius $\sqrt{n}$.
We then turn to understanding a binary constraint, namely where the domain is the discrete hypercube $\Sigma_n =\{+1,-1\}^n$.  Following the statistical physics literature, in the former setting, we call the objective the \emph{spherical} $p$-spin glass Hamiltonian and the latter setting the \emph{Ising} $p$-spin glass Hamiltonian.

In both settings, quite a lot is known about the maximum. It can be shown that the maximum value of $H_n$ has an almost sure limit (as $n \to \infty$ with $p$ fixed), 
called the \emph{ground state energy}, which we will denote by $E_p(\cS)$ 
for the spherical setting and $E_p(\Sigma)$ for the Ising setting, and explicit 
variational formulas are known for $E_p(\cS)$ \cite{ABC13,JagTob17,ChenSen17} and $E_p(\Sigma)$ \cite{Ton02,AuffChen18,JS17}.

Algorithmically, it is known how to find, in polynomial time, a solution of value $E_p^\infty(\cS) - \varepsilon$ or $E_p^\infty(\Sigma_n) - \varepsilon$ (respectively for the spherical and Ising settings) for any constant $\varepsilon > 0$~\cite{subag-full-rsb,montanari-sk,EMS-opt}. In both the spherical and Ising settings, these constants satisfy $E_2^\infty = E_2$ and $E_p^\infty < E_p$ for $p \ge 3$. In other words, it is known how to efficiently optimize arbitrarily close to the optimal value in the $p=2$ case, but not when $p \ge 3$.

\subsubsection{Low-Degree Polynomial Algorithms}\label{sec:poly-alg}

Our goal here is to understand how well one can optimize~\eqref{eq:max-H} via the output of a vector-valued low-degree polynomial in the coefficients $Y$. To simplify notation we will often abuse notation and refer to the space of 
$p$-tensors on $\R^n$ by $\R^m \cong (\R^n)^\tp$ where $m=n^p$.

We say that a function $f:\RR^m\to\RR^n$ is a polynomial of degree (at most) $D$ if it may be written in the form 
\begin{equation}\label{eq:vec-val-poly}
f(Y) = (f_1(Y),\ldots,f_n(Y)),
\end{equation}
where each $f_i:\RR^m\to\R$ is a polynomial of degree at most $D$. 

We will also consider the case where $f$ is allowed to have random coefficients, 
provided that these coefficients are independent of $Y$. That is, we  will assume that 
there is some probability space $(\Omega,\PP_\omega)$ and that 
$f:\RR^m\times\Omega\to\RR^n$ is such that $f(\cdot,\omega)$ is a polynomial of degree 
at most $D$ for each $\omega\in\Omega$. We will abuse notation and refer to this as a \emph{random polynomial} $f: \RR^m \to \RR^n$.

Our precise notion of what it means for a polynomial to optimize~$H_n$ will depend somewhat on the domain~$\cX_n$. This is because it is too much to ask for the polynomial's output to lie in $\cX_n$ exactly, and so we fix a canonical rounding scheme that maps the polynomial's output to $\cX_n$. We begin by defining this notion for the sphere: $\cX_n = \cS_n$.

\paragraph{The spherical case.}

We will round a polynomial's output to the sphere $\cS_n$ by normalizing it in the standard way. To this end, for a random polynomial $f: \RR^m \to \RR^n$ we define the random function $g_f: \RR^m \to \cS_n \cup \{\infty\}$ by
\[
g_f(Y,\omega) = \sqrt n \frac{f(Y,\omega)}{\norm{f(Y,\omega)}_2},
\]
with the convention $g_f(Y,\omega) = \infty$ if $f(Y,\omega)=0$.

\begin{definition}
For parameters $\mu \in \RR$, $\delta \in [0,1]$, $\gamma \in [0,1]$, and a random polynomial $f: \RR^m \to \RR^n$, we say that $f$ $(\mu,\delta,\gamma)$-optimizes the objective \eqref{eq:p-spin-def} on $\cS_n$ if the following are satisfied when $(Y,\omega)\sim\PP_Y\otimes \PP_\omega$:
\begin{itemize}
    \item $\displaystyle \Ex_{Y,\omega} \|f(Y,\omega)\|^2_2 = n$ \; (normalization).
    \item With probability at least $1-\delta$ over $Y$ and $\omega$, we have both $H_n(g_f(Y,\omega);Y) \ge \mu$ and $\|f(Y,\omega)\|_2 \ge \gamma \sqrt{n}$.
\end{itemize}
\end{definition}

\noindent Implicitly in this definition, the case $f(Y,\omega)=0$ must occur with probability at most $\delta$. The meaning of the parameters $(\mu,\delta,\gamma)$ is as follows: $\mu$ is the objective value attained after normalizing the polynomial's output to the sphere, and $\delta$ is the algorithm's failure probability. Finally, $\gamma$ is involved in the norm bound $\|f(Y,\omega)\|_2 \ge \gamma \sqrt{n}$ that we need for technical reasons. Since the domain is $\cS_n$, $f$ is ``supposed to'' output a vector of norm $\sqrt{n}$. While we do not require this to hold exactly (and have corrected for this by normalizing $f$'s output), we do need to require that $f$ usually does not output a vector of norm too much smaller than $\sqrt{n}$. This norm bound is important for our proofs because it ensures that a small change in $f(Y,\omega)$ can only induce a small change in $g_f(Y,\omega)$.

We now state our main result on low-degree hardness of the spherical $p$-spin model, with the proof deferred to Section~\ref{sec:pf-lowdeg-pspin}.

\begin{theorem}\label{thm:spherical-lowdeg}
For any even integer $p \ge 4$ there exist constants $\mu < E_p(\cS)$, $n^* \in \NN$, and $\delta^* > 0$ such that the following holds. For any $n \ge n^*$, any $D \in \NN$, any $\delta \le \min\{\delta^*,\frac{1}{4} \exp(-2D)\}$, and any $\gamma \ge (2/3)^D$, there is no random degree-$D$ polynomial that $(\mu, \delta, \gamma)$-optimizes \eqref{eq:p-spin-def} on $\cS_n$.
\end{theorem}

\noindent A number of remarks are in order. First, this result exhibits a tradeoff between the degree $D$ of polynomials that we can rule out and the failure probability $\delta$ that we need to assume. In order to rule out polynomials of \emph{any} constant degree, we need only the mild assumption $\delta = o(1)$. On the other hand, if we are willing to restrict to algorithms of failure probability $\delta = \exp(-cn)$ (which we believe is reasonable to expect in this setting), we can rule out all polynomials of degree $D \le c'n$ for a constant $c' = c'(c)$. It has been observed in various hypothesis testing problems that the class of degree-$n^\delta$ polynomials is at least as powerful as all known $\exp(n^{\delta-o(1)})$-time algorithms~\cite{sam-thesis,lowdeg-notes,subexp-sparse}. This suggests that optimizing arbitrarily close to the optimal value in the spherical $p$-spin (for $p \ge 4$ even) requires fully exponential time $\exp(n^{1-o(1)})$.

The best known results for polynomial-time optimization of the spherical $p$-spin were first proved by~\cite{subag-full-rsb} but can also be recovered via the AMP framework of~\cite{EMS-opt}. As discussed in Appendix~\ref{app:low-deg-alg}, these AMP algorithms can be captured by constant-degree polynomials. Furthermore, the output of such an algorithm concentrates tightly around $\sqrt{n}$ and thus easily satisfies the norm bound with $\gamma = (2/3)^D$ required by our result. We also expect that these AMP algorithms have failure probability $\delta = \exp(-\Omega(n))$; while this has not been established formally, a similar result on concentration of AMP-type algorithms has been shown by~\cite{gamarnik2019overlap}.

Our results are limited to the case where $p \ge 4$ is even and $\mu$ is a constant slightly smaller than the optimal value $E_p(\cS)$. These restrictions are in place because the OGP property used in our proof is only known to hold for these values of $p$ and $\mu$. If the OGP were proven for other values of $p$ or for a lower threshold $\mu$, our results would immediately extend to give low-degree hardness for these parameters (see Theorem~\ref{thm:spherical-ogp-lowdeg}). Note that we cannot hope for the result to hold when $p=2$ because this is a simple eigenvector problem with no computational hardness: there is a constant-degree algorithm to optimize arbitrarily close to the maximum (see Appendix~\ref{app:low-deg-alg}).

\paragraph{The Ising case.}

We now turn to low-degree hardness in the Ising setting, where the domain is the hypercube: $\cX_n = \Sigma_n$. In this case, we round a polynomial's output to the hypercube by applying the sign function. For $x \in \RR$, let
\[ \sgn(x) = \left\{\begin{array}{ll} +1 & \text{if } x \ge 0 \\ -1 & \text{if } x < 0, \end{array}\right. \]
and for a vector $x \in \RR^n$ let $\sgn(x)$ denote entry-wise application of $\sgn(\cdot)$. We now define our notion of near optimality for a low-degree polynomial.

\begin{definition}
For parameters $\mu \in \RR$, $\delta \in [0,1]$, $\gamma \in [0,1]$, $\eta \in [0,1]$, and a random polynomial $f: \RR^m \to \RR^n$, we say that $f$ $(\mu,\delta,\gamma,\eta)$-optimizes the objective \eqref{eq:p-spin-def} on $\Sigma_n$ if the following are satisfied.
\begin{itemize}
    \item $\displaystyle \Ex_{Y,\omega} \|f(Y,\omega)\|^2_2 = n$ \; (normalization).
    \item With probability at least $1-\delta$ over $Y$ and $\omega$, we have both $H_n(\sgn(f(Y,\omega));Y) \ge \mu$ and \mbox{$|\{i \in [n] \;:\; |f_i(Y,\omega)| \ge \gamma\}| \ge (1 - \eta)n$}. 
\end{itemize}
\end{definition}

\noindent The interpretation of these parameters is similar to the spherical case, with the addition of $\eta$ to take into account issues related to rounding. More precisely, as in the spherical case, $\mu$ is the objective value attained after rounding the polynomial's output to the hypercube, and $\delta$ is the failure probability. The parameters $\gamma, \eta$ are involved in an additional technical condition, which requires $f$'s output not to be too ``small'' in a particular sense. Specifically, all but an $\eta$-fraction of the coordinates of $f$'s output must exceed $\gamma$ in magnitude. The need for this condition in our proof arises in order to prevent a small change in $f(Y,\omega)$ from inducing a large change in $\sgn(f(Y,\omega))$.

We have the following result on low-degree hardness in the Ising setting. The proof is deferred to Section~\ref{sec:pf-lowdeg-pspin}.
\begin{theorem}\label{thm:ising-lowdeg}
For any even integer $p \ge 4$ there exist constants $\mu < E_p(\Sigma)$, $n^* \in \NN$, $\delta^* > 0$, and $\eta > 0$ such that the following holds. For any $n \ge n^*$, any $D \in \NN$, any $\delta \le \min\{\delta^*,\frac{1}{4} \exp(-2D)\}$, and any $\gamma \ge (2/3)^D$, there is no random degree-$D$ polynomial that $(\mu, \delta, \gamma, \eta)$-optimizes \eqref{eq:p-spin-def} on $\Sigma_n$.
\end{theorem}

\noindent This result is very similar to the spherical case, and the discussion following Theorem~\ref{thm:spherical-lowdeg} also applies here. The best known algorithms for the Ising case also fall into the AMP framework~\cite{montanari-sk,EMS-opt} and are thus captured by constant-degree polynomials. These polynomials output a solution ``close'' to the hypercube in a way that satisfies our technical condition involving $\gamma, \eta$. As in the spherical case, the case $p=2$ is computationally tractable; here it is not a simple eigenvector problem but can nonetheless be solved by the AMP algorithm of~\cite{montanari-sk,EMS-opt}.

\subsubsection{Langevin Dynamics and Gradient Descent}

One natural motivation for understanding low-degree hardness is to investigate the performance of natural iterative schemes, such as power iteration or gradient descent. In the spherical $p$-spin model, the natural analogue of these algorithms (in continuous time) are \emph{Langevin dynamics} and \emph{gradient flow}. 
While these are not directly low-degree methods, the overlap gap property can still be seen to imply hardness for these results in a fairly transparent manner. 

To make this precise, let us introduce the following.
Let $B_t$ denote spherical Brownian motion. (For a textbook introduction to spherical Brownian motion see, e.g., \cite{Hsu02}.) For any variance $\sigma\geq 0$, we introduce \emph{Langevin dynamics} for $H_n$ 
to be the strong solution to the stochastic differential equation
\[
dX_t = \sigma dB_t + \nabla H_n(X_t;Y)dt,
\]
with $X_0=x$, where here $\nabla$ denotes the spherical gradient. Note that since $H_n(x;Y)$ is a polynomial in $x$, $H_n$ is (surely) smooth and consequently the solution is well-defined in the strong sense \cite{Hsu02}. The case $\sigma = 0$ is referred to as \emph{gradient flow} on the sphere. 

In this setting, it is natural to study the performance with random starts which are independent of $Y$, e.g., a uniform at random start. In this case, if the initial distribution is given by $X_0\sim\nu$ for some $\nu\in\cM_1(\cS_n)$, the space of probability measures on $\cS_n$, we will denote the law by $Q_\nu$.  In this setting we have the following result which is, again, a consequence of the overlap gap property. 

\begin{theorem}\label{thm:langevin-main}
Let $p\geq 4$ be even. There exists $\mu < E_p(\cS)$ and $c>0$ such that for any $\sigma\geq0$, $T\geq0$ fixed, 
 $n$ sufficiently large, and $\nu\in\cM_1(\cS_n)$, if $X_t$ denotes Langevin dynamics for $H_n(\cdot\,;Y)$ with variance $\sigma$ and initial data $\nu$, then
\[
\PP_Y\otimes Q_\nu(H_n(X_T;Y) \leq \mu )\geq 1-\exp(-c n).
\]
In particular, the result holds for $\nu_n = \mathrm{Unif}(\cS_n)$, the uniform measure on $\cS_n$.
\end{theorem}

\noindent The proof can be found in Section~\ref{sec:pf-langevin}. To our knowledge, this is the first proof that neither Langevin dynamics nor gradient descent reaches the ground state started from uniform at random start. {We note furthermore, that the above applies even to $T \leq c' \log n$ for some $c'>0$ sufficiently small.}

There has been a tremendous amount of attention paid to the Langevin dynamics of spherical $p$-spin glass models. It is impossible here to provide a complete list of references but we point the reader to the surveys \cite{BCKM98,Cug03,Gui07,jagannath2019dynamics}. 
To date, much of the analysis of the dynamics in the \emph{non-activated} regime considered here ($n\to \infty$ and then $t\to\infty$) has concentrated on the  Crisanti--Horner--Sommers--Cugiandolo--Kurchan (CHSCK)  equations approach \cite{crisanti1993sphericalp,CugKur93}.
This approach centers around the analysis 
of a system of integro-differential equations which are satisfied by the scaling limit of natural observables  of the underlying system. While this property of the scaling limit has now been shown rigorously \cite{BADG01,BADG06}, there is limited rigorous understanding of the solutions to the CHSCK equations beyond the case when $p=2$.
A far richer picture is expected here related to the phenomenon of \emph{aging} \cite{Gui07,BA02}. 

More recently a new, differential inequality--based approach to understanding this regime was introduced in \cite{BGJ20}, which provides upper and lower bounds on the energy level reached for a given initial data. That being said, this upper bound is nontrivial only for $\sigma$ sufficiently large.

We end by noting that overlap gap--like properties, namely ``free energy barriers'', have been used to develop spectral gap estimates for Langevin dynamics which control the corresponding $L^2$-mixing time \cite{gheissari2019spectral,arous2018spectral}.  In \cite{arous2018spectral}, it was shown that exponentially-small spectral gaps are connected to the existence of free energy barriers for the overlap, which at very low temperatures can be shown to be equivalent to a variant of the overlap gap property in this setting.  To our knowledge, however, this work is the first approach to connect the behavior of Langevin dynamics in the non-activated regime ($n\to\infty$ and then $t\to\infty$) that utilizes the overlap distribution. Finally we note here that the overlap gap property has been connected to the spectral gap for local, reversible dynamics of  Ising spin glass models in \cite{arous2018spectral} as well as to gradient descent and approximate message passing schemes in \cite{gamarnik2019overlap}.

\subsection{The Maximum Independent Set Problem and Low-Depth Boolean Circuits}

In this section we state our main result regarding superpolynomial size lower bounds on Boolean circuits
with bounded depth for finding large independent sets in graphs. (The conference version of this paper~\cite{gamarnik2020lowFOCS} also included lower bounds against low-degree polynomials for the independent set problem, but we omit these here and instead refer the reader to the subsequent sharper results in~\cite{wein2020optimal}.) While this result is a statement
on the limits of the power of such circuits to find independent sets in any graph, and thus is of interest
from the worst-case perspective, our main technique relies on studying independent
sets in sparse \emph{random} graphs and we begin by providing the necessary background. 

The problem of finding a large independent set in a sparse random graph is defined as follows. 
We are given the adjacency matrix of an $n$-vertex graph, represented as $Y \in \{0,1\}^m$ where $m = \binom{n}{2}$. 
We write $Y \sim \G(n,q_n)$ to denote an \ER graph on $n$ nodes with edge probability $q_n$, 
i.e., every possible edge occurs independently with probability $q_n$. In the following let $d_n = n q_n$.
We are interested in the regime where $d_n$ is order 1 as  $n\to\infty$ but sufficiently large.
A subset of nodes $I\subseteq [n]$ is 
an \emph{independent set} if it spans no edges, i.e., for 
every $i,j \in I$, $(i,j)$ is not an edge. Letting $\cI(Y)$ denote the set of all 
independent sets of the graph $Y$, consider the optimization problem
\begin{equation}\label{eq:max-indep}
\max_{I \in \cI(Y)} |I|,
\end{equation}
where $Y \sim \G(n,q_n)$.

As $n \to \infty$, if $d = n q_n$ is constant, the rescaled optimum value of~\eqref{eq:max-indep} 
is known to converge to some limit with high probability (actually modulo exponentially small in $n$ probability): 
\begin{align*}
    \frac{1}{n}\, {\max_{I \in \cI(Y)} |I|}\to\alpha_d,
\end{align*}
as shown in~\cite{BayatiGamarnikTetali}. 
The limit $\alpha_d$ is known to have the following asymptotic behavior as $d\to\infty$:
\begin{align*}
    \alpha_d=(1+o_d(1)){2\log d\over d},
\end{align*}
as is known since the work of Frieze~\cite{FriezeIndependentSet}. Here $o_d(1)$ denotes a function of $d$ converging
to zero as $d\to\infty$. More precisely, if we let $d_n = n q_n$  and define $\phi_{n,d}=(\log d_n/d_n)n$ for short, then Frieze's bound states that
for every $\epsilon>0$ there exists $\gamma>0$ and large enough $d_0$ such that for all $d_0\leq d_n =O(1)$,
\begin{equation}\label{eq:Max-IS-ER}
\pr\left({{\max_{I \in \cI(Y)} |I|}\over 2\phi_{n,d}}\in (1-\epsilon,1+\epsilon)\right)
\ge 1-\exp(-\gamma n),
\end{equation}
for all large enough $n$.

The fact that the probabilistic guarantee above is of exponential type is of great importance in 
the derivation of our main result. As was explained in the introduction, 
the lack of such guarantees for independent sets in \emph{dense} random graphs is what prevents
us from establishing similar bounds for such graphs.

The best known polynomial-time algorithm for this problem is achieved by a straightforward greedy algorithm~\cite{karp} which constructs a $1/2$-optimal independent set, i.e., an independent set of size $\phi_{n,d}$ asymptotically as $n \to \infty$ and then $d \to \infty$. It is believed that no polynomial-time algorithm can obtain an independent set of size
$(1+\epsilon)\phi_{n,d}$ for any fixed $\epsilon$, again in the doubly asymptotic regime $n\to\infty$, $d\to\infty$.
Our main result below is a step in this direction which shows that such an algorithm cannot be a polynomial-size circuit with a certain bound on its depth.

We thus now turn to the discussion of Boolean circuits.  The  definitions of these can be found in most
standard textbooks on algorithms and computation~\cite{arora2009computational,sipser,alon2004probabilistic}. 
Boolean circuits are functions $C:\{0,1\}^m\to \{0,1\}^n$, for some positive integers $m$ and $n$,
obtained by 
considering a directed graph with $m$ input nodes which have in-degree zero, and $n$ output nodes which have out-degree zero. 
Each intermediate node corresponds to one of three standard Boolean operations 
$\vee, \wedge$ or $\neg$. The fan-in (i.e., the number of inputs for each and/or gate) is unbounded.  
The size of the circuit, denoted by $s(C)$, is the number of nodes in the associated graph and its
depth is the length of its longest directed path. It is known that any circuit can be modified
so that all of the negation gates are at the input layer at the cost of at most doubling the size (see \cite[Chapter 13]{arora2009computational}); 
for technical convenience, it is common in the literature to take the convention that circuits are defined in this latter form, and we adopt this convention here as well. 
In our case we set $m={n\choose 2}$ with 
the input space $\{0,1\}^m$ associated with the adjacency matrices of a graph on $n$ nodes. The value $1$ indicates
the presence of an edge, and the value $0$ indicates the absence of an edge in the underlying graph.
Given  a Boolean circuit $C$ with $m$ input nodes (associated with edges) and $n$ output nodes, and given a graph 
$y\in \{0,1\}^m$, 
we denote by $C(y)$ the binary vector in $\{0,1\}^n$ produced when $C$ takes input $y$. 
The vector $C(y)$ is also associated with a subset of $[n]$ defined by coordinates with value $1$. 
The cardinality of this set (the number of ones) is denoted by $|C(y)|$.

As stated in the introduction, we are interested in Boolean circuits with bounds on their depth which 
provide some approximation guarantees for solving combinatorial optimization problems, the maximum
independent set problem in our case.
Fix any $n$-dependent positive integer-valued sequence $p(n)$ and a constant $0<\rho<1$. 
We denote by $\mathcal{C}(n,p(n),\rho)$ the family of all $m={n\choose 2}$-input,  $n$-output Boolean circuits $C$ which 
satisfy the following properties:

\begin{enumerate}

\item[(a)]
The depth of $C$ is at most $p(n)$.

\item[(b)]
For every graph $y\in \{0,1\}^{m}$, the output $C(G)$ is an independent set in  $y$ 
which satisfies  
\begin{align*}
|C(y)|\ge \rho \max_{I\in \mathcal{I}(y)}|I|.
\end{align*}
\end{enumerate}
In other words, this is a family of circuits which, given a graph as an input,
is required to produce an independent set in this graph with value 
that is at least a multiplicative constant $\rho$ away 
from optimality. We note that this family is non-empty, provided $p(n)$ grows with $n$, as it is possible to implement an exhaustive search type algorithm that finds a largest independent set as a Boolean circuit with depth bound that does not depend on $n$. Our main result below
states that if the depth $p(n)$ is at most roughly $\log n/\log\log n$, any circuit in the family must have super-polynomial size.

\begin{theorem}\label{theorem:Main-ind-set}
Fix any $\rho>1/2$ and $\alpha>0$. 
Let
\begin{align*}
p(n)={\log n\over (1+\alpha)\log\log n}.
\end{align*}
Then for all sufficiently large $n$, and all $C\in \mathcal{C}(n,p(n),\rho)$, 
the size $s(C)$ of the circuit $C$ satisfies
\begin{align*}
s(C)\ge n^{(\log n)^{\alpha/3}}.
\end{align*}
\end{theorem}

The proof of Theorem~\ref{theorem:Main-ind-set} is given in Section~\ref{sec:pf-main}. 
While, as stated, the result pertains to the power of circuits to find large independent sets in arbitrary graphs, the proof is based 
on establishing the failure of circuits to find half-optimum independent sets for sparse \ER graphs 
with large (but finite) average degree (more precisely for $\G(n,d/\g{n})$ where here and in the following, $\g{n}=2^{\lceil\log_2 n\rceil}$).
As in our previous results,  the proof is based on the barriers created by the presence of a multi-OGP.

\begin{remark}
We note that our result essentially refutes all algorithms 
which satisfy item (b) above (namely, they produce $\rho$-optimal solutions for  $\rho>1/2$) 
and whose total influence is bounded by $\sum_i I_i(f)<c n^2$ for some $c$ sufficiently small.
\end{remark}

\subsection{The Overlap Gap and multi-Overlap Gap Property}

As discussed in the introduction, the preceding results will follow due to certain geometric properties of the super-level sets of the objectives. The main property is called the \emph{overlap gap property (OGP)}. Let us begin by defining this formally in a general setting.

\begin{definition}\label{definition:OGP}
We say that a family of real-valued functions $\mathcal{F}$ with common domain $\mathcal{X}\subset \R^n$ satisfies the
(ensemble)-\emph{overlap gap property (OGP)} for an overlap $R:\cX\times \cX\to [0,1]$ with parameters $\mu \in \RR$ and $0\leq\nu_1<\nu_2\leq 1$ if for every $f_1,f_2\in\mathcal{F}$ and every $x_1,x_2\in\cX$ satisfying
$f_k(x_k)\geq \mu$ for $k=1,2$, we have that
$  R(x,y) \in [0,\nu_1]\cup [\nu_2,1].$
\end{definition}
\noindent For ease of notation, when this holds, we simply say that $\mathcal{F}$ satisfies the $(\mu,\nu_1,\nu_2)$-OGP  for $R$ on $\cX$. Furthermore, as it is often clear from context, we omit the dependence of the above on $R$.

The word ``ensemble'' here is in reference to there being a family $\mathcal{F}$ of objective functions.
In many cases the OGP holds for a singleton $\mathcal{F}$, but the proofs often require construction
of a rich enough family (ensemble) of these, as is the case throughout this paper. We will drop the reference to ``ensemble'' from this point on.

While the definition above might be satisfied for trivial reasons and thus not be informative, it will be used in this paper in the setting where $\|x\|_2^2\le n$ for every $x\in \mathcal{X}$, {$R(x_1,x_2)=|\langle x_1,x_2\rangle|/n$},  and with parameters chosen so that with high probability $\mu<\sup_{x\in \mathcal{X}}H(x)$ for every $H\in\mathcal{F}$. Thus, in particular $R(x_1,x_2)\le 1$ for every $x_1,x_2\in \mathcal{X}$, and $\mu$ measures some proximity from optimal values for each objective function $H$. The definition says informally that for every two $\mu$-optimal solutions with respect to any two choices of objective functions, their normalized inner product is either at least $\nu_2$ or at most $\nu_1$. 

In the following, we require one other property of functions, namely separation of their superlevel sets.
\begin{definition}\label{definition:well-separated}
We say that two real-valued functions $f,g$ with common domain $\cX$ are $\nu$-separated above $\mu$ with respect to the overlap $R:\cX\times \cX \to [0,1]$ if for any $x,y \in \cX$ with $f(x)\geq \mu$ and $g(y)\geq \mu$, we have that $R(x,y) \leq \nu$.
\end{definition}
\noindent This property can be thought of a strengthening of OGP for two distinct functions. In particular, the parameter $\nu$ will typically equal the parameter $\nu_1$ in the definition of OGP. 

The definitions above will serve as a basis for our results regarding spherical and ising p-spin optimization problems.
Let us now turn to stating the precise results regarding these properties in the settings we consider here. It can be shown that the overlap gap property holds for $p$-spin glass Hamiltonians in both the spherical and Ising settings with respect to the overlap $R(x,y) =\frac{1}{n}\abs{\g{x,y}}$. More precisely, let $Y \in (\RR^n)^{\otimes p}$ be i.i.d.\ $\mathcal{N}(0,1)$ and let $Y'$ denote an independent copy of $Y$. 
Consider the corresponding family of real-valued functions
\begin{equation}\label{eq:interpolated-family-p-spin}
\cA(Y,Y') =\{\cos(\tau) H_n(\cdot\,;Y)+\sin(\tau)H_n(\cdot\,;Y') \,:\, \tau \in [0,\pi/2]\},
\end{equation}
with $H_n$ defined as in~\eqref{eq:p-spin-def}.  We then have the following, which will follow by combining bounds from \cite{ChenSen17,AuffChen18}.  
The claim below  for the case $\Sigma_n$ can be found as  \cite[Theorem 3.4]{gamarnik2019overlap}. 
The proof for $\mathcal{S}_n$ is given in Section~\ref{sec:pf-ogp-pspin}. Recall from Section~\ref{sec:p-spin} the meaning of $E_p(\mathcal{S})$ and $E_p(\Sigma)$ (ground state energy).

\begin{theorem} \label{thm:pspin-ogp}
Let the  overlap be defined as $R(x,y) =\frac{1}{n}\abs{\g{x,y}}$, 
and let $Y$ and $Y'$ be independent $p$-tensors with i.i.d.\ $\mathcal{N}(0,1)$ entries. 
For every even $p\geq 4$ and for either $\mathcal{X}_n=\mathcal{S}_n$ or $\mathcal{X}_n=\Sigma_n$ (spherical and Ising models), 
there exist $0<\mu<E_p(\mathcal{X})$,  $0\leq\nu_1<\nu_2\leq 1$, and $\gamma>0$ such that the following holds
with probability at least $1-\exp(-\gamma n)$ for all large enough $n$: 

\begin{enumerate}
\item The family  $\mathcal{F}=\cA(Y,Y')$ satisfies the $(\mu,\nu_1,\nu_2)$-OGP. 
\item $H_n(\cdot\,;Y)$ and $H_n(\cdot\,;Y')$ are $\nu_1$-separated above  $\mu$ with respect to $R$.
\end{enumerate}

\end{theorem}

Let us now turn to the maximum independent set problem. 
In this setting  we will need to rely on a more complex version
of the OGP, called multi-OGP, which involves overlaps between more than two solutions. 
 We will not attempt to provide a general abstract definition of the multi-OGP, and instead we state explicitly
the property  we need as a theorem. 
We will see, however, that the multi-OGP described below, when applied to pairs of solutions, boils down to the definition of the OGP given above.

Before stating the multi-OGP let us first recall the Bernoulli encoding of the underlying graph. For discretization reasons, rather than working with the graph  $\G(n,d/n)$,
it is more convenient to work instead with the graph $Y\sim \G(n,q_n)$ with $q_n = d/\g{n}$.
Here, as before, 
$\g{n} = 2^{\lceil\log_2 n\rceil}$ where $\lceil\cdot\rceil$ denotes the ceiling function and $d$ is fixed but large. (Notice that since 
we are proving a \emph{deterministic} refutation result for Boolean circuits we are free to choose the random graph sequence.) Similarly, we may assume that $d$ is an integer.  We do so throughout the rest of the paper.
We view any realization of the graph $Y$ as an $m$-bit string $Y\in \{0,1\}^m$ where $m={n\choose 2}$ as follows.
The distribution of $Y$ is i.i.d.\ Bernoulli, $(\Be(q_n))^{\otimes m}$, 
with success probability $q_n$. 
Our first step is to convert
the input model to the one associated with an unbiased Bernoulli distribution. This will be done using
standard methods for generating any distribution from uniform random seeds.

Consider a sequence of random instances $Y_t$, indexed by integers $t\ge 0$, constructed as follows. First, $Y_0$ is generated from $(\Be(q_n))^{\otimes m}$. Once $Y_t$ is defined, we let $Y_{t+1}$ be obtained from $Y_t$ by resampling 
coordinate $\ell$ of $Y_t$ where $\ell = 1+(t \bmod m)$ from $\Be(q_n)$. 
In other words, $Y_1,\ldots,Y_m$ are obtained from $Y_0$ by resampling coordinates $1,2,\ldots,m$ in that order (where in the graph setting, we have fixed an arbitrary order for the edges).
Then we repeat the process: $Y_{m+1}$ is obtained from $Y_m$ by resampling coordinate $1$, $Y_{m+2}$ from $Y_{m+1}$ by resampling coordinate $2$, etc. Note that $Y_t\sim (\Be(q_n))^{\otimes m}$ for each $t$.
In our case, $Y_t\sim \G(n,q_n)$. It is in the context of this sequence of graphs $Y_t,\, t\ge 0$ that the multi-OGP 
was established earlier in~\cite{wein2020optimal}. 

In this paper, however, for technical reasons we need a slightly modified
``stretched'' version of the multi-OGP established in~\cite{wein2020optimal}. 
Specifically we generate random $\lceil\log_2 n\rceil$-bit 
strings $U_{ij}=U_{ij}(1),\ldots,U_{ij}(\lceil\log_2 n\rceil)$, independently for $1\le i<j\le n$
and uniformly at random in $\{0,1\}^{\lceil\log_2 n\rceil}$. 
Each $U_{ij}$ represents a number in the range $0,1,\ldots,\g{n}-1$ via the standard binary encoding. 
Let $U=(U_{ij})_{1\le i<j\le n}$.
Consider the associated $\{0,1\}-$valued matrix $Y=(Y_{ij})_{1\le i<j\le n}$ where $Y_{ij}=1$ 
if the number represented by $U_{ij}$ is at most $d-1$ and $Y_{ij}=0$ otherwise. That is, $Y_{ij}=1$ iff $\sum_{1\le \ell\le \lceil\log_2 n\rceil}2^{\ell-1} U_{ij}(\ell)\le d-1$. 
Clearly $Y$ is distributed as $\G(n,q_n)$. Let $M={n\choose 2}\lceil\log_2 n\rceil$. 
The associated deterministic function mapping 
$\{0,1\}^M$ to $\{0,1\}^{n\choose 2}$ is denoted by $\Phi$.
That is, $\Phi(U)=Y$. 

Consider now the following sequence  $(U_t)_{t\geq 0}$: 
$U_0$ is generated uniformly at random from $\{0,1\}^M$. Assuming $U_0,\ldots,U_{t-1}$ have been constructed,
$U_t$ is obtained from $U_{t-1}$ by resampling coordinate $\ell$ of $U_{t-1}$ where $\ell = 1+(t \bmod M)$.
Each $U_t$ is distributed as $U$ described above. As a result, each $Y_t=\Phi(U_t)$ is a random graph 
$\G(n,q_n)$. 
If we think of arranging the coordinates of $U$ such that coordinates of $U_{ij}$ appear contiguously, 
we see that the process $(U_t, t\ge 0)$ is indeed a stretched version of $(Y_t, t\ge 0)$ 
where for periods of length $\lceil\log_2 n\rceil$,
a particular edge $(i,j)$ in graph $Y_t$ is toggled according the realizations
of the bits $U_{ij}(1),\ldots,U_{ij}(\lceil\log_2 n\rceil)$, leaving all other edges of the graph $Y_{t}$ intact. 

We now state the modified version of the multi-OGP.  In the following, choose $d_n = n q_n =n d/\g{n}$, 
and recall $\phi_{n,d} = (\log d_n/d_n)n$.

\begin{theorem}[multi-Overlap-Gap-Property~\cite{wein2020optimal}]\label{thm:multi-OGP}
For every $\epsilon>0$, $K\ge 1+5/\epsilon^2$ and $\alpha>0$, there exist $\gamma>0$ and $d_0$ large enough 
such that for all $d\ge d_0$ the following holds with probability at least $1-\exp(-\gamma n)$ 
for all large enough $n$:
there does not exist a sequence of times $t_1, \ldots, t_K$ with $0 \le t_k \le n^\alpha$ 
and sets $I_1,\ldots,I_K\subset [n]$ such that 
\begin{enumerate}
\item[(i)] each $I_k$ is an independent set in the graph $Y_{t_k}=\Phi(U_{t_k})$,
with size $|I_k| \ge (1+\epsilon)\phi_{n,d}$, and

\item[(ii)] for $k = 2,3,\ldots, K$,
\begin{align*}
|I_k\setminus (\cup_{1\le \ell<k} I_\ell)|\in \left[{\epsilon \over 4}\phi_{n,d},{\epsilon \over 2}\phi_{n,d}\right].
\end{align*}

\end{enumerate}

\end{theorem}
The proof of Theorem~\ref{thm:multi-OGP} is the same as that of its earlier version in~\cite{wein2020optimal}
and is based on the first moment argument. The only difference is that the union bound is over a larger 
set of graphs, due to stretching, 
but the multiplicative increase is only $\lceil\log_2 n\rceil$, while the tail bounds are exponential in $n$. 
Thus the result is obtained at the expense of slightly larger prefactor $\gamma$ in the exponent.

If the statement in the theorem above were hypothetically true for $K=2$ and any $\epsilon > 0$,  
the claim would say that the intersection of any two independent sets larger than
half optimum in any two graphs in the family $\{Y_t=\Phi(U_t)\}_{0 \le t \le n^\alpha}$ 
is either at least $(\epsilon/2)\phi_{n,d}$
or at most $(\epsilon/4)\phi_{n,d}$. Namely, it would say that the model exhibits the $(\mu,\nu_1,\nu_2)$-OGP
with $\mu=(1+\epsilon)\phi_{n,d}$, $\nu_1={\epsilon \over 4}\phi_{n,d}$, $\nu_2={\epsilon \over 2}\phi_{n,d}$.
Unfortunately, this appears to be false. While this was not formally verified in the prior literature, 
the method of proof, which is based on the expectation (first moment) argument, does not reveal the presence of such
a gap. It is in fact very likely that it does not exist, and this might be provable by applying the second moment 
method, though we do not pursue this in this paper. However, a such pairwise version of  
OGP does take place if the value $(1+\epsilon)\phi_{n,d}$
is replaced by $(1+1/\sqrt{2}+\epsilon)\phi_{n,d}$, as was shown originally in~\cite{gamarnik2014limits}.
Namely, the OGP is a special case of the multi-OGP when $K=2$, with the right choice of value $\mu$.

This (namely $K=2$) version of the OGP was also employed in the earlier (conference) 
version of this paper~\cite{gamarnik2020lowFOCS}
as a method of ruling out low-degree polynomials as potential algorithms to find independent sets 
of size roughly $(1+1/\sqrt{2})\phi_{n,d}$.  (We note that there, we worked with $\G(n,d/n)$ as opposed to $\G(n,q_n)$ though that does not change the arguments substantially.)
However, in order to improve this to the sharp  
algorithmic threshold $\phi_{n,d}$ we
need to resort to the multi-OGP stated above. 

We now turn to the separation property---the analogue of Definition~\ref{definition:well-separated}.
In fact, we will need a slightly more general version given below, the proof of which is again found in~\cite{wein2020optimal}.

\begin{lemma}\label{lemma:separation-ind-sets}
Fix $\epsilon>0$ and $c>0$. There exists $\gamma>0$ and large enough $d_0$
such that for all $d\ge d_0$,
the following holds for all large enough $n$: for any fixed set $C\subset [n]$ with $|C| \le c\phi_{n,d}$, with probability at least $1-\exp(-\gamma n)$, every independent set $I$ in $\G(n,q_n)$ has $|I\cap C|\le \epsilon\phi_{n,d}$.
\end{lemma}

Intuitively, this lemma states that any  set of nodes with size order $\phi_{n,d}$ which is fixed prior 
to revealing the graph,
has intersection size at most  $\epsilon\phi_{n,d}$ with every independent set in this random graph.

\section{Proofs for $p$-Spin Model}\label{sec:pf-pspin}

\subsection{Low-Degree Polynomials are Stable}

In this section we prove a noise stability--type result for polynomials of Gaussians, which will be a key ingredient in our proofs. Throughout this section, let $d\geq 1$ and let $Y\in \RR^d$ be a vector with i.i.d.\ standard Gaussian entries. Denote the standard Gaussian measure on $\R^d$ by $\Gamma^d$. For two standard Gaussian random vectors defined on the same probability space, we write $X\sim_\rho Y$ if their covariance satisfies $\mathrm{Cov}(X,Y) = \rho\, I$ for some $\rho \in [0,1]$, where $I$ denotes the identity matrix. Throughout this section, all polynomials have non-random coefficients. The goal of this section is to prove the following stability result.
\begin{theorem}
\label{thm:hyp-stable}
Let $0\leq \rho\leq 1$.
Let $X,Y$ be a pair of standard Gaussian random vectors on $\R^d$ such that $X\sim_\rho Y$. Let $P$ denote the joint law of $X, Y$. Let $f: \R^d \to \RR^k$ be a (deterministic) polynomial of degree at most $D$ with $\EE \norm{f(X)}_2^2 = 1$.  For any $t \ge (6e)^D$,
\[ P(\|f(X) - f(Y)\|_2^2 \ge 2t(1-\rho^D)) \le \exp\left(-\frac{D}{3e} t^{1/D}\right). \]
\end{theorem}

We begin by recalling the following standard consequence of hypercontractivity; see Theorem~5.10 and Remark~5.11 of \cite{janson-gaussian} or \cite[Sec.\ 3.2]{LedouxTalagrand}.
\begin{proposition*}[Hypercontractivity for polynomials]
If $f: \R^d \to \RR$ is a degree-$D$ polynomial and $q \in [2,\infty)$ then
\begin{equation}\label{eq:hyp-moment}
\Ex\left[|f(Y)|^q\right] \le (q-1)^{qD/2} \Ex[f(Y)^2]^{q/2}. 
\end{equation} 
\end{proposition*}

\noindent Let us now note the following useful corollary of this result for vector-valued polynomials.

\begin{lemma}
\label{lem:2-norm-moment}
If $f: \RR^d \to \RR^k$ is a degree-$D$ polynomial and $q \in [2,\infty)$ then
\[ 
\Ex[\|f(Y)\|_2^{2q}] \le [3(q-1)]^{qD} \Ex[\|f(Y)\|_2^2]^q.
\]
\end{lemma}
\begin{proof}
Let us begin by observing that by the Cauchy-Schwarz inequality and   \eqref{eq:hyp-moment}, 
\begin{align}
\EE[\|f(Y)\|_2^4] 
&\le \sum_i \EE[f_i(Y)^4] + 2 \sum_{i < j} \sqrt{\EE[f_i(Y)^4]\EE[f_j(Y)^4]}\nonumber \\
&\le \sum_i 9^D \EE[f_i(Y)^2]^2 + 2 \sum_{i < j} 9^D \EE[f_i(Y)^2] \EE[f_j(Y)^2]= 9^D \left(\E \norm{f(Y)}_2^2\right)^2.\label{eq:2-norm-moment-1}
\end{align}
On the other hand, since $\|f(Y)\|_2^2$ is a polynomial of degree at most $2D$, we may again apply \eqref{eq:hyp-moment} to obtain
\[
\EE[\|f(Y)\|_2^{2q}] \le (q-1)^{qD} \EE[\|f(Y)\|_2^4]^{q/2} \leq [3(q-1)]^{qD}\E[\norm{f(Y)}_2^2]^2
\]
as desired, where in the last line we used \eqref{eq:2-norm-moment-1}.
\end{proof}

\noindent With these results in hand we may now prove the following preliminary tail bound. (The case $k=1$ is a classical result; see e.g., Theorem~9.23 of~\cite{o-book}.)
\begin{proposition}
\label{prop:2-norm-tail}
If $f: \R^d \to \RR^k$ is a degree-$D$ polynomial, then for any $t \ge (6e)^D$,
\[ 
\Gamma^d(\|f(Y)\|_2^2 \ge t \,\EE[\|f(Y)\|_2^2]) \le \exp\left(-\frac{D}{3e} t^{1/D}\right). 
\]
\end{proposition}
\noindent (Recall that $\Gamma^d(\cdot)$ denotes probability under standard Gaussian measure.)
\begin{proof}
Using Lemma~\ref{lem:2-norm-moment}, for any $q \in [2,\infty)$,
\begin{align*}
\Gamma^d(\|f(Y)\|_2^2 \ge t) &= \Gamma^d(\|f(Y)\|_2^{2q} \ge t^q) \le \EE[\|f(Y)\|_2^{2q}]t^{-q}\\
& \le [3(q-1)]^{qD} \EE[\|f(Y)\|_2^2]^q t^{-q} 
\le (3q)^{qD}\, \EE[\|f(Y)\|_2^2]^q\, t^{-q}
\end{align*}
and so, letting $q = t^{1/D}/(3e) \ge 2$,
\[ \Gamma^d(\|f(Y)\|_2^2 \ge t \,\EE[\|f(Y)\|_2^2]) \le [(3q)^D/t]^q = \exp(-Dq) = \exp(-D t^{1/D}/(3e)).\qedhere \]
\end{proof}

\noindent It will be helpful to recall the \emph{noise operator}, $T_\rho:L^2(\Gamma^d)\to L^2(\Gamma^d)$, defined by
\[
T_\rho f(x) = \E f(\rho x+\sqrt{1-\rho^2}Y)
\]
where $\rho \in [0,1]$. Recall that for $t\geq 0$, $P_t := T_{e^{-t}}$ is the classical Ornstein-Uhlenbeck semigroup. In particular,
if $(h_\ell)$ are the Hermite polynomials on $\R$ normalized to be an orthonormal basis for $L^2(\Gamma^1)$, then the eigenfunctions of $T_\rho$
are given by products of Hermite polynomials \cite{LedouxTalagrand}.
In particular, for any $\psi(x)$ of the form $\psi(x) = h_{\ell_1}(x_1)\cdots h_{\ell_d}(x_d)$.  
we have
\begin{equation}\label{eq:hermite-eigenval}
T_\rho \psi (x) = \rho^D \psi(x)
\end{equation}
where $D=\sum \ell_j$. With this in hand we are now in position to prove the following inequality. 
 \begin{lemma}\label{lem:ex-noise}
If $f: \RR^d \to \RR^k$ is a degree-$D$ polynomial with $\EE \|f(Y)\|_2^2 = 1$, then for any $\rho \in [0,1]$, if $X\sim_\rho Y$,
\[ 
\Ex \|f(X) - f(Y)\|_2^2 \le 2(1 - \rho^D). 
\] 
\end{lemma}
\begin{proof}
Let $X_\rho$ be given by
\[
X_\rho = \rho Y + \sqrt{1-\rho^2} Y',
\]
where $Y'$ is an independent copy of $Y$. Observe that $(X_\rho,Y)$ is equal in law to $(X,Y)$. In this case, we see that
\begin{align*}
    \Ex\|f(X) - f(Y)\|_2^2 
    = 2 - 2 \EE \langle f(X),  f(Y) \rangle 
    = 2 - 2 \EE \langle f(X_\rho),  f(Y) \rangle
    = 2 - 2 \EE \langle T_\rho f(Y), f(Y) \rangle.
\end{align*}
Consider the collection of products of real valued Hermite polynomials of degree at most $D$,
\[
\mathcal{H}_D =\{\psi:\R^d\to\R \,:\, \psi(x) = h_{\ell_1}(x_1)\cdots h_{\ell_d}(x_d)\; s.t.\, \sum \ell_i \leq D\}.
\]
Observe that $\mathcal{H}_D$ is an orthonormal system in $L^2(\Gamma^d)$ and that the collection of real-valued polynomials $p:\R^d\to\R$ of degree at most $D$ is contained in its closed linear span.  As such, since $\rho^D\leq \rho^s$ for $0\leq s\leq D$, we see that for any $1\leq i \leq d$,
\[
\rho^D \E f_i(Y)^2 \leq \E\, T_\rho f_i(Y) f_i(Y) \leq \E f_i(Y)^2
\]
by \eqref{eq:hermite-eigenval}. Summing in $i$ yields 
\[
\rho^D\leq \E \langle T_\rho f(Y),f(Y)\rangle \leq 1.
\]
Combining this with the preceding bound yields the desired inequality.
\end{proof}

\noindent We are now in position to prove the main theorem of this section.
\begin{proof}[Proof of Theorem~\ref{thm:hyp-stable}]
Let $Y'$ be an independent copy of $Y$. Then 
if we let $\tilde{Y} = (Y,Y')$, this is a standard Gaussian vector on $\R^{2d}$. Furthermore if we let
\[
h(\tilde{Y})= f(Y) - f(\rho Y + \sqrt{1-\rho^2}Y'),
\]
then $h$ is a polynomial of degree at most $D$ in $\tilde{Y}$ and, by Lemma~\ref{lem:ex-noise}, 
\[
\EE \|h(\tilde Y)\|_2^2=\E\norm{f(X)-f(Y)}_2^2  \le 2(1-\rho^D).
\]
The result now follows from Proposition~\ref{prop:2-norm-tail}.
\end{proof}

\subsection{Failure of Low-Degree Algorithms}\label{sec:pf-lowdeg-pspin}

In this section we prove our main results on low-degree hardness for the spherical and Ising $p$-spin models (Theorems~\ref{thm:spherical-lowdeg} and \ref{thm:ising-lowdeg}). The main content of this section is to show that the OGP and separation properties imply failure of stable algorithms, following an interpolation argument similar to~\cite{gamarnik2019overlap}. The main results then follow by combining this with the stability of low-degree polynomials (Theorem~\ref{thm:hyp-stable}) and the fact that OGP and separation are known to hold (Theorem~\ref{thm:pspin-ogp}).

\paragraph{The spherical case.}

We begin by observing the following elementary fact: when two vectors of norm at least $\gamma$ are normalized onto the unit sphere, the distance between them can only increase by a factor of $\gamma^{-1}$.
\begin{lemma}\label{lem:norm-bound}
If $\|x\|_2 = \|y\|_2 = 1$ and $a \ge \gamma$, $b \ge \gamma$ then $\|x - y\|_2 \le \gamma^{-1} \|ax - by\|_2$.
\end{lemma}
\begin{proof}
 We have
 \[ \|ax - by\|_2^2 = a^2 + b^2 - 2ab \langle x,y \rangle = (a-b)^2 + ab \|x - y\|_2^2 \ge \gamma^2 \|x - y\|_2^2.\qedhere \]
\end{proof}

\noindent Throughout the following, it will be convenient to define the following interpolated family of tensors. Consider $(Y_\tau)_{\tau \in [0,\pi/2]}$ defined by
\begin{equation}\label{eq:Y-tau-def}
Y_\tau = \cos(\tau) Y+ \sin(\tau) Y'.
\end{equation}
Note that by linearity of inner products, we may equivalently write $\cA(Y,Y')$ from \eqref{eq:interpolated-family-p-spin} as
\[
\cA(Y,Y') =\{ H_n(x;Y_\tau) \,:\, \tau \in [0,\pi/2]\}.
\]

\noindent The following result shows that together, the OGP and separation properties imply failure of low-degree polynomials for the spherical $p$-spin.

\begin{theorem}\label{thm:spherical-ogp-lowdeg}
For any $0 \le \nu_1 < \nu_2 \le 1$, there exists a constant $\delta^* > 0$ such that the following holds. Let $p, n, D \in \NN$ and $\mu \in \RR$.
Suppose that $Y,Y'$ are independent $p$-tensors with i.i.d.\ standard Gaussian entries and let $\cA(Y,Y')$ be as in \eqref{eq:interpolated-family-p-spin}. Suppose further that with probability at least $3/4$ over $Y,Y'$, we have that $\cA(Y,Y')$ has the $(\mu,\nu_1,\nu_2)$-OGP on domain $\cS_n$ {with overlap $R=|\langle \cdot,\cdot\rangle|/n$}, and that $H_n(\cdot\,,Y)$ and $H_n(\cdot\,,Y')$ are $\nu_1$ separated above $\mu$. Then for any $\delta \le \min\{\delta^*,\frac{1}{4} \exp(-2D)\}$ and any $\gamma \ge (2/3)^D$, there is no random degree-$D$ polynomial that $(\mu, \delta, \gamma)$-optimizes \eqref{eq:p-spin-def} on $\cS_n$.
\end{theorem}

\begin{proof}
Let $Y,Y'$ be as in the statement of the theorem, and let $P=\PP_Y \otimes \PP_\omega$ denote the joint law of $(Y,\omega).$
Assume on the contrary that $f$ is a random degree $D$ polynomial which $(\mu,\delta,\gamma)$-optimizes $H_n(\cdot\,;Y)$. We first reduce to the case where $f$ is deterministic. 

Let $A(Y,\omega)$ denote the ``failure'' event
\[
A(Y,\omega)= \{H_n(g_f(Y,\omega);Y) < \mu \;\vee\; \|f(Y,\omega)\|_2 < \gamma \sqrt{n}\}.
\]
Since $\EE \|f(Y,\omega)\|_2^2 = n$ and $\PP(A(Y,\omega)) \le \delta$, we have by Markov's inequality, 
\[
\PP_\omega\{\EE_Y \|f(Y,\omega)\|_2^2 \ge 3n\} \le 1/3
\quad \text{ and }\quad 
\PP_\omega(\PP_Y(A(Y,\omega)) \ge 3\delta) \le 1/3.
\]
This means that there exists an $\omega^* \in \Omega$ such that $\EE_Y \|f(Y,\omega^*)\|_2^2 \le 3n$ and $\PP_Y\{A(Y,\omega^*)\} \le 3\delta$. Fix this choice of $\omega = \omega^*$ so that $f(\cdot) = f(\cdot,\omega^*)$ becomes a deterministic function.

Let $Y,Y' \in (\RR^n)^\tp$ be independently i.i.d.\ $\mathcal{N}(0,1)$, let $Y_\tau$ be as in \eqref{eq:Y-tau-def}, and $\cA(Y,Y')$ as in \eqref{eq:interpolated-family-p-spin}.
For some $L \in \NN$ to be chosen later, divide the interval $[0,\pi/2]$ into $L$ equal sub-intervals: $0 = \tau_0 < \tau_1 < \cdots < \tau_L = \pi/2$, and let $x_\ell = g_f(Y_{\tau_\ell})$. We claim that with positive probability (over $Y,Y'$), all of the following events occur simultaneously and that this leads to a contradiction:
\begin{enumerate}
    \item [(i)] The family $\cA(Y,Y')$ has the  $(\mu,\nu_1,\nu_2)$-OGP on $\cS_n$ and $H_n(\cdot\,;Y)$ and $H_n(\cdot\,; Y')$ are $\nu_1$-separated above $\mu$.
    \item [(ii)] For all $\ell \in \{0,1,\ldots,L\}$, $f$ succeeds on input $Y_{\tau_\ell}$, i.e., the event $A(Y_{\tau_\ell},\omega^*)^c$ holds.
    \item[(iii)] For all $\ell \in \{0,1,\ldots,L-1\}$, $\|f(Y_{\tau_\ell}) - f(Y_{\tau_{\ell+1}})\|_2^2 < \gamma^2 cn$ for some $c = c(\nu_1,\nu_2) > 0$ to be chosen later.
\end{enumerate}
First, let us see why (i)-(iii) imply a contradiction. Combining (i) and (ii) gives $|\frac{1}{n} \langle x_0,x_\ell \rangle| \in [0,\nu_1] \cup [\nu_2,1]$ for all $\ell$, and $|\frac{1}{n} \langle x_0,x_L \rangle| \in [0,\nu_1]$. Since we also have $|\frac{1}{n} \langle x_0,x_0 \rangle| = 1$, there must exist an $\ell$ that crosses the OGP gap in the sense that
\[ 
\nu_2 - \nu_1 \le \frac{1}{n}\Big|\abs{\g{x_0,x_\ell}} - \abs{\g{x_0,x_{\ell+1}}}\Big| \leq
\frac{1}{n}|\langle x_0,x_\ell \rangle - \langle x_0,x_{\ell+1} \rangle| \leq  \frac{1}{\sqrt n} \|x_\ell - x_{\ell+1}\|_2.
\]
Since $\|f(Y_{\tau_\ell})\|_2, \|f(Y_{\tau_{\ell+1}})\|_2 \ge \gamma \sqrt{n}$ by (ii), Lemma~\ref{lem:norm-bound} gives
\[ \nu_2 - \nu_1 \le \frac{1}{\sqrt n} \|x_\ell - x_{\ell+1}\|_2 \le \frac{1}{\gamma \sqrt{n}} \|f(Y_{\tau_\ell}) - f(Y_{\tau_{\ell+1}})\|_2, \]
which contradicts (iii) provided we choose $c \le (\nu_2-\nu_1)^2$.

It remains to show that (i)-(iii) occur simultaneously with positive probability. By assumption, (i) fails with probability at most $1/4$, so it is sufficient to show that (ii) and (iii) each fail with probability at most $1/3$. By a union bound, (ii) fails with probability at most $3\delta (L+1)$, which is at most $1/3$ provided
\begin{equation}\label{eq:L-cond-1}
L \le \frac{1}{9\delta} - 1.
\end{equation}
For (iii), we will apply Theorem~\ref{thm:hyp-stable} with some $\tilde D \ge D$ (since we are allowed to use any upper bound on the degree) and $t = (6e)^{\tilde D}$. For any $\ell$ we have $Y_{\tau_\ell} \sim_\rho Y_{\tau_{\ell+1}}$ with $\rho = \cos\left(\frac{\pi}{2L}\right)$. Using $\EE_Y \|f(Y)\|_2^2 \le 3n$,
\begin{equation}\label{eq:ffd}
\PP( \|f(Y_{\tau_\ell}) - f(Y_{\tau_{\ell+1}})\|_2^2 \ge 6n(6e)^{\tilde D}(1-\rho^{\tilde D}) ) \le \exp(-2 \tilde{D}).
\end{equation}
Since
\[ 1-\rho^{\tilde D} = 1-\cos^{\tilde D}\left(\frac{\pi}{2L}\right) \le 1 - \left(1 - \frac{1}{2}\left(\frac{\pi}{2L}\right)^2\right)^{\tilde D} \le \frac{\tilde D}{2}\left(\frac{\pi}{2L}\right)^2, \]
equation~\eqref{eq:ffd} implies
\[ \PP( \|f(Y_{\tau_\ell}) - f(Y_{\tau_{\ell+1}})\|_2^2 \ge \gamma^2 cn ) \le \exp(-2\tilde{D}) \]
provided
\begin{equation}\label{eq:L-cond-2}
L \ge \frac{\pi}{2\gamma} \sqrt{\frac{3\tilde{D}}{c}}(6e)^{\tilde D/2}.
\end{equation}
Thus, (iii) fails with probability at most $L \exp(-2\tilde{D})$, which is at most $1/3$ (as desired) provided
\begin{equation}\label{eq:L-cond-3}
L \le \frac{1}{3} \exp(2\tilde{D}).
\end{equation}
To complete the proof, we need to choose integers $\tilde D \ge D$ and $L$ satisfying~\eqref{eq:L-cond-1}, \eqref{eq:L-cond-2}, \eqref{eq:L-cond-3}, i.e.,
\begin{equation}\label{eq:L-final}
\frac{\pi}{2\gamma} \sqrt{\frac{3\tilde{D}}{c}}(\sqrt{6e})^{\tilde D} \le L \le \min\left\{\frac{1}{9\delta}-1, \frac{1}{3}(e^2)^{\tilde D}\right\}.
\end{equation}
Require $\delta \le \frac{1}{4} \exp(-2\tilde{D})$ so that the second term in the $\min\{\cdots\}$ is smaller (when $\tilde{D}$ is sufficiently large). Since $\gamma \ge (2/3)^D \ge (2/3)^{\tilde D}$ and $\frac{3}{2}\sqrt{6e} < e^2$, there now exists an $L \in \NN$ satisfying~\eqref{eq:L-final} provided that $\tilde D$ exceeds some constant $D^* = D^*(c)$. Set $\tilde D = \max\{D,D^*\}$ and $\delta^* = \frac{1}{4} \exp(-2D^*)$ to complete the proof.
\end{proof}

\noindent Our main result on low-degree hardness of the spherical $p$-spin now follows by combining the above with the fact that OGP and separation hold in a neighborhood of the optimum.

\begin{proof}[Proof of Theorem~\ref{thm:spherical-lowdeg}]
This result follows by combining Theorem~\ref{thm:spherical-ogp-lowdeg} with Theorem~\ref{thm:pspin-ogp}.
\end{proof}

\paragraph{The Ising case.} We now turn to the corresponding result for the Ising $p$-spin model, which again shows that together, OGP and separation imply failure of low-degree polynomials. 

\begin{theorem}\label{thm:ising-ogp-lowdeg}
For any $0 \le \nu_1 < \nu_2 \le 1$ there exist constants $\delta^* > 0$ and $\eta > 0$ such that the following holds. Let $p, n, D \in \NN$ and $\mu \in \RR$. Suppose that $Y,Y'$ are independent $p$-tensors with i.i.d.\ standard Gaussian entries and let $\cA(Y,Y')$ be as in \eqref{eq:interpolated-family-p-spin}. Suppose further that with probability at least $3/4$ over $Y,Y'$, we have that $\cA(Y,Y')$ has the $(\mu,\nu_1,\nu_2)$-OGP on domain $\Sigma_n$ with overlap $R=|\langle \cdot,\cdot\rangle|/n$, and that $H_n(\cdot\,;Y)$ and $H_n(\cdot\,;Y')$ are $\nu_1$ separated above $\mu$. Then for any $\delta \le \min\{\delta^*,\frac{1}{4} \exp(-2D)\}$ and any $\gamma \ge (2/3)^D$, there is no random degree-$D$ polynomial that $(\mu, \delta, \gamma, \eta)$-optimizes \eqref{eq:p-spin-def} on $\Sigma_n$.
\end{theorem}
\begin{proof}
The proof is nearly identical to that of Theorem~\ref{thm:spherical-ogp-lowdeg} above, so we only explain the differences. We now define $A(Y,\omega)$ to be the failure event 
\[A(Y,\omega) = \{H_n(\sgn(f(Y,\omega));Y) < \mu \;\vee\; |\{k \in [n] \;:\; |f_k(Y,\omega)| \ge \gamma\}| < (1 - \eta)n\},
\]
and define $x_\ell = \sgn(f(Y_{\tau_\ell}))$. The only part of the proof we need to modify is the proof that (i)-(iii) imply a contradiction, including the choice of $c$. As above, combining (i) and (ii) gives the existence of an $\ell$ for which $\nu_2 - \nu_1 \le \frac{1}{\sqrt n} \|x_\ell - x_{\ell+1}\|_2$, i.e., $\frac{1}{4}\|x_\ell - x_{\ell+1}\|_2^2 \ge \frac{1}{4}(\nu_2 - \nu_1)^2 n$, implying that $x_\ell$ and $x_{\ell+1}$ differ in at least $\Delta := \frac{1}{4}(\nu_2 - \nu_1)^2 n$ coordinates. Let $\eta = \Delta/(2n) = \frac{1}{8}(\nu_2 - \nu_1)^2$ so that there must be at least $\Delta/2$ coordinates $i$ for which $|f_i(Y_{\tau_\ell}) - f_i(Y_{\tau_{\ell+1}})| \ge \gamma$. This implies $\|f(Y_{\tau_\ell}) - f(Y_{\tau_{\ell+1}})\|_2^2 \ge \gamma^2 \cdot \frac{\Delta}{2} = \frac{1}{8}\gamma^2 (\nu_2 - \nu_1)^2 n$, which contradicts (iii) provided we choose $c \le \frac{1}{8}(\nu_2 - \nu_1)^2$.
\end{proof}

\begin{proof}[Proof of Theorem~\ref{thm:ising-lowdeg}]
This result follows by combining Theorems~\ref{thm:ising-ogp-lowdeg} and \ref{thm:pspin-ogp}.
\end{proof}

\subsection{Stability of Langevin and Gradient Flows}
Let $U\in C^\infty(\cS_n)$ be some smooth function and  for any $\sigma\geq0$ we can consider \emph{Langevin Dynamics with potential $U$ and  variance $\sigma$} to be the strong solution of the stochastic differential equation
\[
\begin{cases}
dX_t = \sigma dB_t - \nabla U dt\\
X_0 \sim \nu,
\end{cases}
\]
where $B_t$ is spherical Brownian motion, $\nabla$ is the spherical gradient, and $\nu\in\mathcal{M}_1(\cS_n)$ is some probability measure on the sphere called \emph{the initial data}. Note that in the case $\sigma=0$ this is simply gradient flow for $U$.

We recall here the following basic fact about the well-posedness of such equations, namely their continuous dependence on the function $U$. In the following, for a vector-valued function $F: \cS_n\to T\cS_n$, we let $\norm{F}_\infty$ denote the essential supremum of the norm of $F$ induced by the canonical metric. (Here $T\cS_n$ denotes the tangent bundle to $\cS_n$.)
\begin{lemma*}
Let $U,V\in C^\infty(\cS_n)$ and  $\sigma\geq 0$.  Fix $\nu\in \mathcal{M}_1(\cS_n)$. Let  $X^U_t$  and $X^V_t$ denote the corresponding solutions to Langevin dynamics with potentials $U$ and $V$ respectively and with the same variance $\sigma$ with respect to the same Brownian motion $B_t$. Suppose further that their initial data are the same. Then there is a universal $C>0$ such that for any $t>0$
\begin{equation}\label{eq:well-posed}
\sup_{s\leq t}\norm{X^U_s - X^V_s}_2 \leq C t e^{Ct \norm{\nabla^2 U}_\infty \vee \norm{\nabla^2 V}_\infty}\norm{\nabla U-\nabla V}_\infty \quad \text{ a.s.,}
\end{equation}
where $\norm{\cdot}_2$ denotes Euclidean distance in the canonical embedding of $\cS_n\subseteq \RR^n$.
\end{lemma*}
\noindent The proof of this result is a standard consequence of Gronwall's inequality and can be seen, e.g., in \cite{varadhan,Tes12}.

 In this section, for a $p$-tensor $A$ we will write $A(x_1,\cdots,x_p)$ to denote the action of $A$ on $p$ vectors, i.e., $A(x_1,\ldots, x_p)= \g{A,x_1\otimes\cdots\otimes x_p}.$ Viewing this as a multilinear operator, we denote the operator norm by
\[
\norm{A}_{\op} = \sup_{\norm{x_1}_2=\cdots=\norm{x_p}_2=1} A(x_1,\ldots,x_p),
\]
and for a function $f$ we let $\norm{\nabla^k f}_\infty = \sup_x \abs{\nabla^k f}_{\op}(x)$. Let $Y$ denote a $p$-tensor with standard Gaussian entries. Then, by a standard $\eps$-net argument 
(see, e.g., \cite[Lemma 3.7]{BGJ20}), we have
\begin{equation}\label{eq:gaussian-op-norm}
\begin{aligned}
\E \norm{Y}_{\op} \leq C \sqrt{n}\qquad \text{ and }\qquad
\norm{Y}_{\op} \leq C' \sqrt n \quad  \text{ w.p. } 1-\exp(-\Omega(n)),
\end{aligned}
\end{equation}
for some $C,C'>0$. 
As a consequence of the above, we note the following. 

\begin{lemma}\label{lem:langevin-stability-main}
Let $\delta = n^{-\alpha}$ for some {fixed } $\alpha>0$ and let $\{\tau_i\}$ denote a partition of $[0,\pi/2]$ with $|\tau_{i+1}-\tau_i| \leq \delta$ with $\lceil\delta^{-1}\rceil+1$ elements.  Let $(X^{\tau_i})_i$  denote the family of strong solutions to Langevin dynamics with variance $\sigma\geq0$,  potentials $H_n(\cdot\,;Y_{\tau_i})$ and initial data, $\nu\in\mathcal{M}_1(\cS_n)$. We have that there is a $C>0$ independent of $n$ such that for any $T>0$ 
\[
\sup_{i}\sup_{s\leq T}\norm{X_s^{\tau_i}-X_s^{\tau_{i+1}}}_2 \leq C T e^{C T/n} n^{-\alpha+1/2}
\]
with probability at least $1-e^{-\Omega(n)}$.
\end{lemma}
\begin{proof}
Evidently, the proof will follow by \eqref{eq:well-posed} upon controlling the gradient and Hessian of $H_n(\cdot\,;Y)$.
 To this end,  we see that
\[
\nabla H_n(x;Y_\tau) = \frac{1}{n^{\frac{p+1}{2}}}(Y_{\tau}(\pi_x,x,\ldots,x)+\cdots+Y_{\tau}(x,\ldots,x,\pi_x))
\]
where $\pi_x$ denotes the projection onto $T_x \cS_n$. In particular,
\[
\norm{\nabla H_n(x;Y_{\tau})}_2 \leq \frac{p}{n} \norm{Y_{\tau}}_{\op} \leq \frac{p}{n}(\norm{Y}_{\op}+\norm{Y'}_{\op}).
\]
By a standard epsilon-net argument (see, e.g., \cite[Lemma 3.7]{BGJ20}), we have by applying \eqref{eq:gaussian-op-norm} and a union bound,
\[
\sup_{0\leq\tau\leq \pi/2} \norm{\nabla H_n(\cdot\,;Y_\tau)}_\infty \leq 2C/\sqrt{n},
\]
and, by similar reasoning,
\[
\sup_{0\leq\tau\leq \pi/2} \norm{\nabla^2 H_n(\cdot\,;Y_\tau)}_\infty \leq 2C/n
\]
for some $C>0$ with probability $1-\exp(-\Omega(n))$.

On the other hand, in law we have that
$Y_{\tau_i}-Y_{\tau_{i+1}} = Z$ satisfies 
\[
Z \sim Y\sqrt{(\cos(\tau_i)-\cos(\tau_{i+1}))^2+(\sin(\tau_i)-\sin(\tau_{i+1}))^2}.
\]
Since both cosine and sine are 1-Lipschitz, we see that the entries of $Z$ are i.i.d.\ and have variance at most $\delta$. 
Consequently, by the same epsilon-net argument, we have with probability $1-O(n^\alpha e^{-c n})$,
\[
\max_{i}\norm{\nabla H_n(\cdot\,;Y_{\tau_i})-\nabla H_n(\cdot\,;Y_{\tau_{i+1}})}_2\leq C \delta/\sqrt{n}
\]
as desired.
\end{proof}

\subsection{Failure of Langevin Dynamics}\label{sec:pf-langevin}

We begin by noting the following useful concentration argument.
\begin{lemma}\label{lem:kirszbraun-conc}
 Let $f_d:\R^d\to\R$ be a sequence of Borel functions and $E_d\subseteq \R^d$ a sequence of Borel sets, and let $Z$ be a standard Gaussian vector in $\R^d$.
Suppose that there are some $\alpha,\beta>0$ and $K,C,c>0$ such that: $\abs{f_d(0)},\E f_d(Z)^2 = O(d^\alpha)$, 
$f_d$ is $K$-Lipschitz on $E_d$,
and $P(Z\in E_d^c)\leq C\exp(-cd^\beta).$
Then there are some $C',c'>0$ such that for any $t>0$, we have
\[
P(\abs{f_d(Z)-\E f_d(Z)}>t)\leq C'\exp\{-c'd(t^2/K)\wedge d^\beta\},
\]
for $d$ sufficiently large.
\end{lemma}
\begin{proof}
In the following, the constants $C,c>0$ may change from line to line.
Let
\[
g_d(x) = \inf_{y\in E_d} \{ f_d(y) + K\norm{x-y}_2\}.
\]
Since $f_d$ is $K$-Lipschitz, so is $g_d$. (This can be checked directly or by Kirszbraun's extension theorem.)
By Gaussian concentration for Lipschitz functions (see, e.g., \cite{vershynin2018high}), we then have 
\[
P(\abs{g_d(Z)-\E g_d(Z)}>t)\leq Ce^{-c d t^2/K}.
\]
By construction, we have $\abs{g_d(x)}\leq \abs{f(0)}+ K \norm{x}$, so that by assumption,
\[
\E g_d(Z)^2 \leq \abs{f_d(0)}+\E\norm{x}= O(d^\alpha).
\]
As such, we have that 
\[
\E\abs{f_d(Z)-g_d(Z)}\leq \sqrt{\E f_d^2 + \E g_d^2}\cdot P(E_d^c) = O(d^\alpha \exp(-cd)) = e^{-\Omega(d)}\leq t/2.
\]
for $d$ large enough.
Combining these bounds, we obtain
\begin{align*}
P(\abs{f-\E f} \geq t)&\leq P(\abs{g-\E f}\geq t, E_d) + P(E_d^c) \leq P(\abs{g-\E g}\geq t/2) + P(E^c)\leq C(e^{-cdt^2/K}+e^{-cd}),
\end{align*}
so that
\[
P(\abs{f-\E f} \geq t)\leq C e^{-cdt^2/K\wedge 1}.
\]
as desired
\end{proof}

\begin{lemma}\label{lem:concetration-langevin}
Let $T \leq L N $ for some $L>0$ and $\sigma\geq0$. Let $X_T$ denote the solution of Langevin dynamics with potential $H_n$, variance $\sigma$, and initial data $\nu\in\cM_1(\cS_n)$, and let $Q_\nu$  denote its law conditionally on $Y$. Then we have that there are some $c,c'>0$ such that for every $\eps>0$ and $n$ sufficiently large,
\[
\PP_Y\otimes Q_\nu( \abs{H_n(X_T;Y)-\EE H_n(X_T;Y)} \geq \eps) \leq \exp\left(-c n {\frac{\eps^2}{{Le^{c'L}}}}\right).
\]
\end{lemma}
\begin{proof}
We follow here the same proof as in \cite{gamarnik2019overlap}, modified for the continuous-time setting and making the time-dependence explicit. Recall as above that for any two tensors $Y$ and $Y'$, we have
\[
\norm{\nabla H_n(\cdot\,;Y)-\nabla H_n(\cdot\,;Y')}_2\leq \frac{1}{n}\norm{Y-Y'}_{\op} \leq \frac{1}{n}\norm{Y-Y'}_2,
\]
where here for a tensor $A$, $\norm{A}_2$ denotes the square root of the sum of the squares of its entries,
Recall further that that $\norm{\nabla^2 H_n(\cdot\,;Y)}_{\infty} \leq \norm{Y}_{\op}\, n^{-3/2}\leq C/n$ for some $C>0$
with probability $1-\exp(-\Omega(n))$, where we have used here \eqref{eq:gaussian-op-norm}.
Thus if we let $d=n^{p}$, and consider the event $E_d = \{\norm{Y}_{\op} \geq C \sqrt{n}\}$, we have that $P(E_d)= e^{-\Omega{(d^{1/p})}}$. Furthermore we have by \eqref{eq:well-posed} that the map $F (Y)= X_T$ is $C(T/\sqrt{n})e^{T/n}$-lipschitz.
On this event, we then have for $T \leq L N$, that
\[
\abs{H_n(F(Y);Y)-H_n(F(Y');Y')}\leq \frac{\norm{Y}_{\op}}{n}\norm{F(Y)-F(Y')}_2 +\frac{1}{\sqrt{n}}\norm{Y-Y'}_{\op}
\leq (CLe^{L}+\frac{1}{\sqrt{n}})\norm{Y-Y'}_2. 
\]
Thus the map $f_d(Y) = H_n(F(Y);Y)$ is $CLe^{L}$-lipschitz on $E_d$. Furthermore $f_d(0)=0$ and 
$$\E f_d(x)^2 \leq \E(\max_{x\in\cS_n} H_n(x;Y))^2 \leq \E\left(\frac{1}{\sqrt{n}}\norm{Y}_{\op}\right) \leq C$$
by \eqref{eq:gaussian-op-norm}.
The result then follows by Lemma~\ref{lem:kirszbraun-conc} and the fact that $n=d^{1/p}< d$.
\end{proof}

\begin{proof}[Proof of Theorem~\ref{thm:langevin-main}]
In the following, we let $P=\PP\otimes Q_\mu$. Recall the family $(Y_\tau)$ from \eqref{eq:Y-tau-def} and $\cA(Y,Y')$ from \eqref{eq:interpolated-family-p-spin}. Let  $\delta=n^{-\alpha}$ for some $\alpha>0$ and define $(\tau_i)$ as in Lemma~\ref{lem:langevin-stability-main}.
Fix an $\eps>0$ and let $G$ denote the event that the overlap gap property holds for $\cA(Y,Y')$ with parameters $(E_p(\cS)-\eps,\nu_1,\nu_2)$
as well as $\nu_1$-separation of $H_n(\cdot\,;Y_0)$ and $H_n(\cdot\,;Y_1)$ above level $E_p(\cS)-\eps$. By Theorem~\ref{thm:pspin-ogp}, this holds for every $\eps>0$ sufficiently small with probability $1-\exp(-\Omega(n))$.

Let $X^{\tau_i}$ denote the  solutions to Langevin dynamics
corresponding to the potentials $H_n(\cdot\,;Y_{\tau_i})$.
Let $B_n$ and $\tilde{B}_n$ denote  the bad events
\begin{align*}
\tilde{B}_n &= \{\exists i \,:\, H_n(X_T^{\tau_i};Y_{\tau_i}) \geq E_p(\cS)-\eps \}\\
B_n &= \{ H_n(X_T^{\tau_i};Y_{\tau_i}) \geq E_p(\cS)- 3\eps \;\, \forall i\}.
\end{align*}
Let $E_i(\eps)$ denote the complement of the event bounded in Lemma~\ref{lem:concetration-langevin} applied to $X^{\tau_i}_T$, and let $E(\eps)= \cap E_i(\eps)$ which has probability at least $1-\exp(-\Omega(n))$. Note that on $\tilde{B}_n\cap E(\eps)$, we have that $\E H_n(X^{\tau_i}_T;Y_{\tau_i})\geq E_p(\cS)-2\eps$ for some $i$. As the expectation is non-random and independent of $i$, this holds for all $i$. Consequently, $\tilde{B}_n\cap E(\eps)\subset B_n$. Thus we have $P(\tilde{B}_n)\leq P(B_n) + \exp(-\Omega(n))$. 

Suppose now that the events $B_n$ and $G$ have non-empty intersection. Let us work on this intersection.
By $\nu_1$-separation, recalling the overlap function $R(x,y)=\abs{\frac{1}{n}\g{x,y}}$, we have that 
\[
R(X^0_T,X^1_T) \leq \nu_1
\]
whereas $R(X^0,X^0)=1$. On the other hand, by Lemma~\ref{lem:langevin-stability-main},  it follows that
\[
\abs{R(X_T^0,X_T^{\tau_i})-R(X_T^0,X_T^{\tau_{i+1}})}\leq C T e^{C T/n}n^{-\alpha}.
\]
Thus we see that for $n$ sufficiently large, there must be some (random) $j$ such that 
\[
\nu_1 < \abs{R(X_T^0,X_T^{\tau_j})} <\nu_2.
\]
This contradicts the overlap gap property. Thus $B_n\subseteq G^c$.
Consequently, we have that 
\[
P(\tilde{B}_n)\leq P(B_n) +e^{-\Omega(n)} \leq P(G^c)+e^{-\Omega(n)}=e^{-\Omega(n)}.
\]
Observing that $\tilde{B}_n^c$ is contained in the event we are trying to bound yields the desired result by monotonicity of probabilities.
\end{proof}

\subsection{Proof of Overlap Gap Property}
\label{sec:pf-ogp-pspin}

\begin{proof}[Proof of Theorem~\ref{thm:pspin-ogp}]
We begin with the spherical setting. Let us view $H_n(x;Y)$ as a Gaussian process on $\cS_n$. It was shown in \cite[Theorem 3]{ChenSen17} that  for any $\tau,\eps>0$ with 
$\tau \leq \pi/2$ there are
$C,c,\tilde \mu>0$ such that with probability at least $1-Ce^{-c n}$, 
\[
\max_{R(x,y)>\eps}  H_n(x;Y_\tau) + H_n(x,Y)
<\max_{x\in\cS_n} H_n(x;Y_\tau)+\max_{x\in\cS_n} H_n(x;Y)-  \tilde \mu,
\]
so that if both $u,v$ satisfy
\begin{equation*}
    \begin{aligned}
        \max_{x\in\cS_n} H_n(x;Y_\tau) - \tilde \mu /2 &\leq H_n(u,Y_\tau)\\
        \max_{x\in\cS_n} H_n(x;Y) - \tilde \mu /2 &\leq H_n(v;Y)
    \end{aligned}
\end{equation*}
then it must be that $|R(u,v)|<\eps$. (The result is stated there for $\tau <\pi/2$, but can be extended to the easier case of $\tau = \pi/2$. See Remark~\ref{rem:tau} below.) One can then replace the maximum on the right-hand side of the above upon recalling that by the Borell--TIS inequality, 
\[
\PP( \abs{ \max_{x\in \cS_n} H_n(x;Y) -\E \max_{x\in\cS_n} H_n(x;Y)} \geq \eps)\leq C\exp(-cn\eps^2)
\]
for some $C,c>0$. In particular, upon recalling that $\E\max_{\cS_n} H_n(x;Y)\to E_p(\cS)$ \cite{JagTob17}, for $n$ sufficiently large  we obtain
\begin{equation}\label{eq:disorder-overlap}
    \begin{aligned}
        E_p(\cS) - \tilde \mu /4 &\leq H_n(u,Y_\tau)\\
        E_p(\cS) - \tilde \mu /4 &\leq H_n(v;Y).
    \end{aligned}
\end{equation}
On the other hand, as shown in \cite[Theorem 6]{AuffChen18}, \eqref{eq:disorder-overlap} holds with $\tau =0$ as well, except now we have that the inner products of the near-maximal $u,v$ must satisfy $R(u,v)\in[0,\nu_1]\cup[\nu_2,1]$ for some $0\leq\nu_1<\nu_2\leq1$. By combining these results we can obtain the overlap gap 
property with parameters $(E_p(\cS)-\tilde\mu/4,\nu_1,\nu_2)$ 
by applying the discretization argument from in 
\cite{gamarnik2019overlap}.  Note that, 
\eqref{eq:disorder-overlap} in the case $\tau = \pi/2$ implies 
$\eps$-separation below level $E_p(\cS)-\tilde \mu/4$. As $\eps$ was arbitrarily small we can take $\eps=\nu_1$.

After recalling that $\E \max_{\Sigma_n} H_n(x;Y)\to E_p(\Sigma)$ \cite{Ton02}, we see that the second result is a restatement of \cite[Theorem 3.4]{gamarnik2019overlap} after applying Borell's inequality as  in \eqref{eq:disorder-overlap}.
\end{proof}
\begin{remark}\label{rem:tau}
While the result of \cite[Theorem 3]{ChenSen17} is only stated for $0<\tau<\pi/2$, it easily extends to the case $\tau = \pi/2$ by differentiating in the Lagrange multiplier term $\lambda$ in the ``RSB bound'' from \cite[Eq.\ 59]{ChenSen17}. 
For the reader's convenience, we sketch this change. We follow here the notation of \cite{ChenSen17}. By comparing to \cite[Eq.\ 78]{ChenSen17}, one sees that $E(0,u,\lambda)$ from \cite[Eq.\ 61]{ChenSen17} satisfies $E(0,u,0)= 2 E_p(\cS) \, (= 2GS )$. On the other hand for $u>0$ we have $\partial_\lambda E(0,u,0) = - u<0$, from which it follows that $\min_\lambda T(0,u,\lambda)<2 E_p(\cS)$ as desired. The case $u<0$ follows by symmetry.
\end{remark}

\section{Proofs for the Maximum Independent Set Problem}

\subsection{Random Walks Avoiding Marked Edges}
\label{sec:rand-walk}

This section is devoted to proving a key fact about random walks on the hypercube (Proposition~\ref{prop:no-jump-events}), which may be of independent interest. We remark that this result can be generalized to the $p$-biased hypercube (see the conference version of this paper~\cite{gamarnik2020lowFOCS}, and~\cite{wein2020optimal}), but we will not need this here.

Consider the hypercube graph on vertex set  $\{\pm 1\}^m$, where vertices $y$ and $y'$ are connected by an edge when they differ in precisely
one coordinate. Denote the edge set by $E_m$. In the following, we call a $\{0,1\}$-valued function on the edge set, 
$\phi:E_m\to \{0,1\}$, a \emph{marking} of the edges $E_m$, or simply a marking. Here we view the set of edges 
where $\phi = 1$ as ``marked''. 
We extend $\phi$ to self-loops $(y,y), y\in \{\pm 1\}^m$ by $\phi(y,y)=0$ for all $y$.

We now consider an interpolated family of samples from $\{\pm 1\}^m$ similar to the one underlying
Theorem~\ref{thm:multi-OGP}.
Fix an integer $r \ge 1$ and consider a process $Y_0,Y_1,\ldots,Y_{r}$, constructed as follows.
$Y_0$ is generated uniformly at random from $\{\pm 1\}^m$. Once $Y_0,\ldots,Y_{t}$ are defined,
$Y_{t+1}$ is obtained from $Y_{t}$ by resampling coordinate $\ell$ of $Y_{t}$ where $\ell = 1+(t \bmod m)$. Put simply, $Y_1,\ldots,Y_m$ are obtained by resampling the coordinates $1,\ldots,m$ 
one-by-one, but $Y_{m+1}$ is obtained from $Y_m$ by resampling coordinate $1$, $Y_{m+2}$ is 
obtained from $Y_{m+1}$ by resampling coordinate $2$, and so on.  
Finally, in the following, let 
\[
\cE_t = \{ \phi(Y_{t-1},Y_t) = 1\} 
\]
denote the event that $(Y_{t-1},Y_t)$ is an  edge (that is, $Y_{t-1}\ne Y_t$) 
and that this edge is marked.

\begin{proposition}\label{prop:no-jump-events}
For any marking function $\phi$ and any integer $r \ge 1$,  the following holds:
\begin{align*}
\pr\left(\cap_{t\in [r]}\,\mathcal{E}_t^c\right)\ge \exp\left(-2\log 2\sum_{t\in [r]}\pr\left(\mathcal{E}_t\right)\right).
\end{align*}
\end{proposition}

\begin{proof}
For every $y\in \{\pm 1\}^m$, we let $p_s(y)=\pr(\cap_{t \leq s}\,\mathcal{E}_t^c\,|\,Y_0=y)$. 
The main step in proving the claim is the following
relation.

\begin{lemma}\label{lemma:qx-bound}
For any marking function $\phi$ and any $r\geq 1$,  the following bound holds:
\begin{equation}\label{eq:qx-bound}
-\E \log p_r(Y) \le 2\log 2\sum_{t\leq r}\pr(\mathcal{E}_t),
\end{equation}
where $Y\in \{\pm 1\}^{m}$ is uniformly random. 
\end{lemma}
\noindent From this lemma, the proof of the proposition is immediate: by Jensen's inequality
\begin{align*}
\log \pr\left(\cap_{t\in [r]}\,\mathcal{E}_t^c\right) &= \log \E_Y[p_r(Y)] \ge \E_Y[\log p_r(Y)] \ge -2\log 2\sum_{t\in [r]}\pr(\mathcal{E}_t).\qedhere
\end{align*}
\end{proof}

\noindent We now focus on establishing Lemma~\ref{lemma:qx-bound}. 

\begin{proof}[Proof of Lemma~\ref{lemma:qx-bound}] 
The proof is by induction on $r$. Consider the base case $r=1$, 
for which we claim that the inequality~\eqref{eq:qx-bound} is in fact an equality. 
Indeed, conditioning on $Y_0=y$, suppose first that 
changing the value of $y_1$ (the first coordinate of $y$) to $-y_1$ 
does not traverse a marked edge. Namely, for $y'=(-y_1,y_2,\ldots,y_m)$, the edge $(y,y')$ is not marked.
In this case, $p_r(y)=\pr(\mathcal{E}_1^c\,|\,Y_0=y)=1$, which means both $\log p_r(y)$ and $\pr(\mathcal{E}_1\,|\,Y_0=y)$ are equal to $0$,
i.e.,~\eqref{eq:qx-bound} holds with equality when conditioned on $Y_0=y$.

On the other hand, if the switch from $y_1$ to $-y_1$
does traverse a marked edge, then $p_r(y)=\pr(\mathcal{E}_1^c\,|\,Y_0=y)=1/2$ and $\pr(\mathcal{E}_1\,|\,Y_0=y)=1/2$ 
(as the switch occurs with probability $1/2$). Thus 
\begin{align*}
\log p_r(y)&=\log(1/2) = -2\log 2 \,\pr(\mathcal{E}_1\,|\,Y_0=y),
\end{align*}
and the identity verifies again. The claim for the base case $r=1$ now follows by averaging over $Y_0$.

We now prove the induction step.
Assume the assertion holds for $r'\le r-1$. 
We now establish it for $r$. For any  $y\in \{\pm 1\}^{m}$ let $p_{1,r}(y)$ 
be the probability that the 
path $Y_1,Y_2,\ldots,Y_r$ (note that it starts from $Y_1$ not $Y_0$) 
does not contain any marked edges $(Y_{t-1},Y_{t}),\, 2\le t\le r$, 
when  conditioned on $Y_1=y$. Observe that by our inductive assumption we have for uniformly random $Y \in \{\pm 1\}^m$,
\begin{align}\label{eq:1-to-m}
-\E \log p_{1,r}(Y) \le 2\log 2 \sum_{2\le t\le r}\pr(\mathcal{E}_t).
\end{align}

For any $y\in \{\pm 1\}^{m}$, let $y_{\pm}$ be obtained from $y$ by forcing coordinate
$1$ to be $\pm 1$ (respectively). For $Y \in \{\pm 1\}^m$ uniformly random, we immediately have
\begin{align}\label{eq:qm-expanded}
\E \log p_r(Y)= \EE \left[(1/2)\log p_r(Y_+)+(1/2)\log p_r(Y_-)\right].
\end{align}
Let $\mathcal{F}$ be the event that flipping coordinate $1$ of $Y$ traverses a marked edge, i.e., the event that $(Y_+, Y_-)$ is marked. 
Note that this
event is measurable (determined) by the realization of $Y$. Also note that $\pr(\mathcal{E}_1)=(1/2)\pr(\mathcal{F})$, since $Y_0$ has the same distribution as $Y$, and the likelihood of a swap (i.e., $Y_1 \ne Y_0$) is $1/2$.

We claim
\begin{align}\label{eq:J0+}
p_r(Y_+)=(1/2)p_{1,r}(Y_{+})+(1/2)p_{1,r}(Y_{-})\mb{1}(\mathcal{F}^c).
\end{align}
We justify this identity as follows. Conditioned on $Y_0 = Y_+$, we have that with probability $1/2$, $Y_1 = Y_+$. In this case, no marked edge is encountered at the first step, and the probability that no marked edges are encountered during the remaining $r-1$ steps is $p_{1,r}(Y_{1,+})$. On the other hand, with probability $1/2$ we have $Y_1 = Y_-$. In this case no marked edges are encountered provided that $\mathcal{F}^c$ occurs (so that the edge $(Y_+,Y_-)$ traversed on the first step is unmarked) and furthermore no marked edges are 
encountered during the remaining $r-1$ resamplings, the probability of which is $p_{1,r}(Y_{-})$. 

Similarly, 
\begin{align}\label{eq:J0-}
p_r(Y_{-})=(1/2)p_{1,r}(Y_{-})+(1/2)p_{1,r}(Y_{+})\mb{1}(\mathcal{F}^c).
\end{align}
If the event $\mathcal{F}^c$ occurs, the right-hand sides of the
expressions (\ref{eq:J0+}) and (\ref{eq:J0-}) are identical, so 
on this event we obtain by the concavity of $\log$
\begin{align*}
(1/2)\log p_r(Y_{+})+(1/2)\log p_r(Y_{-})&=\log\left((1/2)p_{1,r}(Y_{+})+(1/2)p_{1,r}(Y_{-})\right) \\
&\ge (1/2)\log\left(p_{1,r}(Y_{+})\right)+(1/2)\log\left(p_{1,r}(Y_{-})\right).
\end{align*}
On the other hand, if the event $\mathcal{F}$ occurs then
\begin{align*}
\log p_r(Y_{+})=\log\left((1/2)p_{1,r}(Y_{+})\right),
\end{align*}
and
\begin{align*}
\log p_r(Y_{-})=\log\left((1/2)p_{1,r}(Y_{-})\right).
\end{align*}
In this case,
\begin{align*}
(1/2)\log p_r(Y_{+})+(1/2)\log p_r(Y_{-})&=
(1/2)\log\left((1/2)p_{1,r}(Y_{+})\right)+(1/2)\log\left((1/2)p_{1,r}(Y_{-})\right) \\
&=-\log 2+(1/2)\log\left(p_{1,r}(Y_{+})\right)+(1/2)\log\left(p_{1,r}(Y_{-})\right).
\end{align*}
Combining, we obtain
\begin{align*}
&(1/2)\log p_r(Y_{+})+(1/2)\log p_r(Y_{-}) \ge -(\log 2)\mb{1}(\mathcal{F})+(1/2)\log\left(p_{1,r}(Y_{+})\right)+(1/2)\log\left(p_{1,r}(Y_{-})\right).
\end{align*}

\noindent Applying (\ref{eq:qm-expanded}) we obtain
\begin{align*}
-\E \log p_r(Y) &\le (\log 2)\pr(\mathcal{F})-(1/2)\E \log p_{1,r}(Y_{+})-(1/2)\E \log p_{1,r}(Y_{-}) \\
&=(\log 2)\pr(\mathcal{F})-\E \log p_{1,m}(Y).
\end{align*}
Recalling $\pr(\mathcal{F})=2\pr(\mathcal{E}_1)$ and applying (\ref{eq:1-to-m})
we obtain the claim.
\end{proof}

\subsection{Total Influence and the Likelihood of No Large Jumps}

{In this section we will specialize the result from the previous section to a particular choice of ``marked'' edges. } We begin by recalling several notions from Boolean/Fourier analysis on $\{\pm 1\}^m$; 
see e.g.~\cite{o-book} for a reference. In our setting, we will have $m={n\choose 2}$.
Consider the standard Fourier expansion of functions on $\{\pm 1\}^{m}$ associated with the
uniform measure on $\{\pm 1\}^{m}$. The basis for this expansion is composed of monomials of the form
$y_S := \prod_{t\in S}y_t,\, S\subset [m]$. For every function $g:\{\pm 1\}^{m}\to \R$, the associated
Fourier coefficients are 
\begin{align*}
\hat g_S=\E[g(Y) Y_S], \qquad S\subset [m],
\end{align*}
where the expectation is with respect to the uniform measure on $Y\in  \{\pm 1\}^{m}$. 
Then the Fourier expansion of $g$ is 
$g(y)=\sum_{S\subset [m]} \hat g_S y_S$, and the Parseval (or Walsh) identity states that $\sum_S \hat g_S^2=\E[g(Y)^2]$. 
For $i \in [m]$, let $L_i$ denote the Laplacian operator:
\begin{align}\label{eq:L-in-Fourier}
L_i g(y) =\sum_{S\ni i}\hat g_S y_S,
\end{align}
and let 
$I(g)$ be the total influence
\begin{align}\label{eq:I-in-L}
I(g)= \sum_{i\in  [m]} \E[(L_i g(Y))^2]  
=\sum_{S \subset [m]} |S|\, \hat g_S^2. 
\end{align}
Also note that
\begin{align}
\E[(L_i g(Y))^2]={1\over 2}\left(\E[(L_i g(Y_{-i,-1}))^2]+\E[(L_i g(Y_{-i,1}))^2]\right) \label{eq:L-i-conditioned}
\end{align}
where for $\ell=\pm 1$, $y_{-i,\ell}\in \{\pm1\}^{m}$ is obtained from $y$ by fixing
the $i$-th coordinate of $y$ to $\ell$. 

Fix any function $g:\{\pm 1\}^m\to \R^n$ and $\kappa\in (0,1)$. We write $g = (g_1,\ldots,g_n)$, 
with each $g_j:\{\pm 1\}^m\to \R$.
We call and edge $(y,y')\in E_n$ (as defined in Section~\ref{sec:rand-walk}) $\kappa$-marked if 
\begin{equation}\label{eq:def-marked}
\|g(y')-g(y)\|_2^2\ge \kappa\, \E[\|g(Y)\|_2^2],
\end{equation}
where the expectation is over uniformly random $Y \in \{\pm 1\}^m$. 
Thus, the event of encountering no marked
edges in the process $(Y_t)_{0 \le t \le Tm}$ (defined in Section~\ref{sec:rand-walk}) amounts to saying that all distances $\|g(Y_t)-g(Y_{t-1})\|_2^2$
are at most $\kappa\, \E[\|g(Y)\|_2^2]$. 
As before, we let $\mathcal{E}_t$ denote
the event that  $Y_{t-1}\ne Y_t$ and the edge $(Y_{t-1},Y_t)$ is marked.
Our next goal is bounding the likelihood of encountering no marked edges in terms of the total influence.

\begin{theorem}\label{thm:no-jumps-vs-influence}
For any $g\in \{\pm 1\}^m\to \R^n$, $\kappa\in (0,1)$, and integer $T\geq 1$, the probability of encountering no $\kappa$-marked edges in the 
sequence $Y_0,\ldots,Y_{Tm}$ satisfies
\begin{align*}
\pr\left(\cap_{t\in [Tm]}\,\mathcal{E}_t^c\right)\ge \exp\left(-{4(\log 2) T\over\kappa\, \E\left[\|g(Y)\|_2^2\right]} 
\sum_{j\in [n]}I(g_j)\right)
\end{align*}
where the expectation is over uniformly random $Y \in \{\pm 1\}^m$.
\end{theorem}

\begin{proof}
We have
\begin{align*}
\sum_{j\in [n]}I(g_j)&=\sum_{j\in [n]}\sum_{i\in [m]} \E[(L_i g(Y))^2]
={1\over 2}\sum_{j\in [n]}\sum_{i\in [m]}
\left(\E[(L_{i} g_j(Y_{-{i},-1}))^2]+\E[(L_{i} g_j(Y_{-{i},1}))^2]\right).
\end{align*}
Using the inequality $a^2+b^2\ge (1/2)(a-b)^2$, the right-hand side is at least
\begin{align*}
\sum_{j\in [n]}\sum_{i\in [m]}
{1\over 4}\left(\E[(L_{i} g_j(Y_{-{i},-1})-L_{i} g_j(Y_{-{i},1}))^2]\right),
\end{align*}
Using  (\ref{eq:L-in-Fourier}),
\begin{align*}
L_{i} g_j(Y_{-{i},-1})-L_{i} g_j(Y_{-{i},1})=g_j(Y_{-{i},-1})-g_j(Y_{-{i},1}),
\end{align*}
implying
\begin{align*}
\sum_{j\in [n]}I(g_j)&\ge {1\over 4}\sum_{j\in [n]}\sum_{i\in [m]}
\E\left[\left(g_j(Y_{-{i},-1})-g_j(Y_{-{i},1})\right)^2\right] \\
&={1\over 4}\sum_{i\in [m]}
\E\left[\|g(Y_{-{i},-1})-g(Y_{-{i},1})\|_2^2\right] \\
&\ge {1\over 2}\kappa\,\E[\|g(Y))\|_2^2]\sum_{i\in [m]}\pr\left(\mathcal{E}_i\right),
\end{align*}
where in the last step we used the definition of ``$\kappa$-marked''~\eqref{eq:def-marked}, the definition of $\mathcal{E}_t$, and the fact that 
the event $Y_t\ne Y_{t-1}$ occurs with
probability $1/2$; in particular, the edge $(Y_{-i,-1},Y_{-i,1})$ is marked with probability $2 \pr(\mathcal{E}_i)$.  
Next we observe that 
\begin{align*}
\sum_{t\in [Tm]}\pr\left(\mathcal{E}_t\right)=T\sum_{i\in [m]}\pr\left(\mathcal{E}_i\right).
\end{align*}
Applying Proposition~\ref{prop:no-jump-events}, we obtain the claim.
\end{proof}

\subsection{Encoding Bit Representation of Graphs into Circuits}
Recall the bit representation $U$ of the random graph $Y\sim \G(n,q_n)$  where $q_n = d/\g{n}$ 
through the function $\Phi(U)=Y$. 
Let $b_0,b_1,\ldots,b_{d-1}$ be the $\lceil\log_2 n\rceil$-bit binary encoding of the numbers $0,1,\ldots,d-1$,
so that each $Y_{ij}=1$ if $U_{ij}$ is one of $b_0,\ldots,b_{d-1}$ and $Y_{ij}=0$ otherwise.
We expand each such $b_r$ as $b_r(1),\ldots,b_r(\lceil\log_2 n\rceil)$.

Our next step is verifying that the function $\Phi$ can be represented by a depth-3 polynomial-size Boolean circuit
in a rather straightforward way.

\begin{lemma}\label{lemma:uniform-to-bias-circuit}
The function $\Phi$ can be represented by a depth-$3$ polynomial-size (in $n$) Boolean circuit.
\end{lemma}

\begin{proof}
We begin by making simplifying observations regarding general Boolean circuits. 
Suppose that we need to encode a function of the form $x \mapsto {\bf 1}(x=c)$, where $c$ (either zero or one) is a fixed constant. Namely, the function outputs one when $x=c$ and zero otherwise.
Then more specifically, when $c=0$, this function takes value $1$ when $x=0$ and value $0$ when $x=1$.
This is simply $\neg x$. In the case $c=1$, the function ${\bf 1}(x=c)$ is simply $x$. It will be convenient to
encode gates with functions of the form ${\bf 1}(x=c)$ without expanding them by cases $c=0$ and $c=1$.

We describe the required circuit $C$ by layers and gates connecting the layers. The input
layer consists of $M={n\choose 2}\lceil\log_2 n\rceil$  nodes associated with ${n\choose 2}$ segments, each length $\lceil\log_2 n\rceil$.
For $1 \le i < j \le n$, we denote the nodes associated with the  $(i,j)$-the segment by $u_{ij}(1),\ldots,u_{ij}(\lceil\log_2 n\rceil)$. 
With some abuse of notation we think of $u_{ij}(\cdot)$ also as Boolean input variables.
Note that $M$ is also the number of input variables.

The number of output nodes is ${n\choose 2}$ and the nodes are denoted by $y_{ij},\, 1\le i<j\le n$. 
Each output variable $y_{ij}$ is computed by the following Boolean function of the input variables:

\begin{align*}
\bigvee_{0\le r\le d-1}
\left(\bigwedge_{\ell \in \lceil\log_2 n\rceil}\mb{1} \left(u_{ij}(\ell)=b_r(\ell)\right)\right).
\end{align*}
The resulting output in the output layer is then exactly $y=\Phi(u)$ for every $u\in \{0,1\}^{m}$. This circuit has depth 3 because there are 4 layers of nodes: input, output, and two intermediate (one to compute the ``and'' and one for the ``or''). The total number of nodes is (broken down by layer) $M + d (\lceil\log_2 n\rceil) \binom{n}{2} + d \binom{n}{2} + \binom{n}{2}$, which is polynomial in $n$ (for any fixed $d$).
\end{proof}

\subsection{Proof of Theorem~\ref{theorem:Main-ind-set}}
\label{sec:pf-main}

Fix any circuit $C\in \mathcal{C}(n,p(n),\rho)$. We compose this circuit with the circuit $\Phi$
constructed in Lemma~\ref{lemma:uniform-to-bias-circuit} to obtain a new circuit $C\circ \Phi$ with depth at most $p(n)+3$ acting on strings
in $\{0,1\}^{M}$. Recall the notation $M={n\choose 2}\lceil\log_2 n\rceil$. We study the size $s(C\circ \Phi)$  of this combined
circuit $C\circ \Phi$. It suffices to establish the following lower bound on $s(C\circ \Phi)$:
\begin{align}
s(C {\circ \Phi})\ge n^{(\log n)^{\alpha/2}}, \label{eq:alpha-over-2}
\end{align} 
as the  size of the circuit $\Phi$ is only polynomial in $n$.

We assume without loss of generality that the input space for $C$ is $\{\pm 1\}^{M}$,
as opposed to $\{0,1\}^M$, 
for consistency with the Fourier-analytic notation earlier. 
The notion of edges in $\{\pm 1\}^{M}$
is the same as before: $(u,u')$ is an edge if $u$ and $u'$ differ in exactly one coordinate. Fix 
$\kappa>0$, the actual value of which will be set below.
Consider now the following marking.
An edge $(u,u')$ is defined to be \emph{$\kappa$-marked}  if
\begin{align*}
\|C\circ \Phi(u)-C\circ \Phi(u')\|_2^2 &\ge \kappa\, \E\left[\|C\circ \Phi(U)\|_2^2\right].
\end{align*}

Fix a positive integer constant $T$ and consider the sequence  $U=U_0,U_1,\ldots,U_{TM}$, constructed as earlier
for random graphs but now for the uniform distribution on $\{\pm 1\}^M$ instead. 
We recall it for convenience: $U_0$ is generated uniformly at random from $\{\pm 1\}^M$. Once $U_0,\ldots,U_{t-1}$ are defined,
$U_t$ is obtained from $U_{t-1}$ by resampling coordinate $\ell$ of $U_{t-1}$ where $\ell = 1+(t \bmod M)$. A key step to establishing our main result is the following implication of Theorem~\ref{thm:no-jumps-vs-influence}.

\begin{theorem}\label{theorem:no-bad}
For every $d$ there exists $c=c(d)>0$ such that the following holds. For any $\rho,\kappa,T$, and  $\eta>0$,
the probability that the sequence $(U_t)_{0\le t\le TM}$ encounters no $\kappa$-marked  edges (that is, no pair
$(U_{t-1},U_t),\, t\in [TM]$ is a $\kappa$-marked edge)
is at least
\begin{align}\label{eq:rate}
\exp\left(-\frac{c T}{\rho \kappa }\left(\log s(C\circ \Phi)\right)^{(1+\eta)p(n)}\right),
\end{align}
for all large enough $n$ and any $C\in\cC(n,p(n),\rho)$.
\end{theorem}

Note that this probability converges to zero since $p(n)\ge 1$. For our purposes, we will need the rate 
of convergence in this bound to be no faster than exponential. In particular, this requirement will control our depth bound.
Before proving this result, let's first see how it implies Theorem~\ref{theorem:Main-ind-set}.

\begin{proof}[Proof of Theorem~\ref{theorem:Main-ind-set}]  
Fix $\rho>1/2$ and let $\epsilon>0$ be such that $\rho=(1+\epsilon)/(2(1-\epsilon))$. Fix 
 any circuit $C\in \mathcal{C}(n,p(n),\rho)$. 
Denote the output produced by the circuit, $C\circ \Phi$, acting on 
$U_t$ by $I_t, 0\le t\le TM$, each viewed either as a vector in  $\{0,1\}^n$ or as a subset of $[n]$.
Recall Theorem~\ref{thm:multi-OGP} and
fix $T=\lceil 1+5/\eps^2\rceil$ as in that theorem. By Frieze's bound \eqref{eq:Max-IS-ER}, there exists $d'_1,\gamma_1>0$ such that for all $d\ge d'_1$ and all large enough $n$, with probability at least $1-(TM+1)\exp(-\gamma_1 n)$, the maximal objective values associated with instances $U_t$ (or more precisely, the ones associated with graphs $Y_t=\Phi(U_t)$) are  all at least
$(1-\eps)2\phi_{n,d}$ and at most $(1+\eps)2\phi_{n,d}$.
Applying Lemma~\ref{lemma:separation-ind-sets} for 
\begin{align}\label{eq:a}
c=2T(1+\epsilon), 
\end{align}
there exists $d'_2,\gamma_2>0$ such that for all $d\ge d'_2$ 
\begin{align}\label{eq:intersection-union}
\max_{t\le (T-1)M,R\subset \{0,\ldots,t\},|R|=T} |I_{t+M}\cap \left(\cup_{r\in R}I_r\right)|
\le (\epsilon/4)\phi_{n,d}.
\end{align}
with probability at least $1-\exp(-\gamma_2 n)$ for all large enough $n$. This holds 
since the instance
$U_{t+M}$ is  independent from the instances $U_0,\ldots,U_t$ for every $t\le TM-M$.

Finally, by  Theorem~\ref{thm:multi-OGP},  there exists a large enough $d'_3$ and $\gamma_3$  such that for 
all large enough $n$, 
with probability at least $1-\exp(-\gamma_3 n)$ the m-e-OGP holds for all $d\ge d'_3$ 
for our choice of $\epsilon$ and $T=K$ and the sequence $U_t, 0\le t\le TM$. 
(Note that we use the theorem in the regime where the length of the sequence is $O(n^2\log n)$, whereas it holds for any
sequence of length $n^{O(1)}$.)

Take $d_0=\max(d'_1,d'_2,d'_3)$ and $\gamma<\min(\gamma_1,\gamma_2,\gamma_3)$
so that for $d\geq d_0$, all of these three events hold  
with probability at least $1-\exp(-\gamma n)$ for all large enough $n$. 
We have used here that $T$ is a constant and $M$ is only polynomially large in $n$.
Call this event $\mathcal{A}$ 
so that when $d\ge d_0$, we have that
$\pr(\mathcal{A})\ge 1-\exp(-\gamma n)$ for all large enough $n$. 

Let $\kappa = \eps/{12}$ and denote the event described in Theorem~\ref{theorem:no-bad}---namely the event that 
no pair $(U_{t-1},U_{t}), t\in [TM]$ is a $\kappa$-marked edge---by $\mathcal{B}$.
Choose $\eta > 0$ small enough so that 
\begin{align}\label{eq:eta}
{1+\alpha\over 1+\eta}> 1+{\alpha/2}.
\end{align}
By Theorem~\ref{theorem:no-bad} applied to this choice of $\eta$,
since $\rho,c,\kappa,T$ are constants, we have
\begin{align*}
\pr(\mathcal{B}) \ge \exp\left(-\left(\log s(C\circ \Phi)\right)^{(1+\eta) p(n)}\right),
\end{align*}
for all large enough $n$.

We claim that $\mathcal{B}\subset \mathcal{A}^c$. Assuming this claim, we obtain for any {$\eta>0$},
\begin{align*}
\exp\left(-\left(\log s(C\circ \Phi)\right)^{(1+\eta) p(n)}\right) \le \pr(\mathcal{B}) 
\le \pr(\mathcal{A}^c) \le \exp(-\gamma n),
\end{align*}
for all large enough $n$, implying
\begin{align*}
(1+\eta) p(n)\log\log s(C\circ \Phi)\ge  \log n+\log \gamma.
\end{align*}
In light of the assumed depth bound 
\begin{align*}
p(n)\le {\log n\over (1+\alpha)\log\log n},
\end{align*}
and using  (\ref{eq:eta})
we obtain a superpolynomial (\ref{eq:alpha-over-2}) as claimed.

It remains to prove the claim $\mathcal{B}\subset \mathcal{A}^c$. The proof is by contradiction. 
Suppose that $\mathcal{B}\cap \mathcal{A}$ is non-empty.
Note that on $\mathcal{A}$, the largest independent is of
size at least $2(1-\epsilon)\phi_{n,d}$ and at most $2(1+\epsilon)\phi_{n,d}$
for every graph $\Phi(U_t)$.
By our assumptions  on the  family $\mathcal{C}$, 
the outputs $I_t$ of the circuit $C\circ \Phi$ acting on $U_t$ are independent sets 
with size at least $\rho \cdot 2(1-\epsilon)\phi_{n,d}=(1+\epsilon)\phi_{n,d}$. (We have used here our
definition  of $\epsilon$.)

On the event $\mathcal{A}$ we have
$\|C\circ \Phi(U)\|_2^2 \leq 2(1+\epsilon)\phi_{n,d}$ for $d\ge d_0$ and $n$ large enough. Since $\pr(\mathcal{A})\ge 1-\exp(-\gamma n)$ and $\|C\circ \Phi(U)\|_2^2 \leq n$, this implies that $\E[\|C\circ \Phi(U)\|_2^2] \le 3 \phi_{n,d}$ for $d\ge d_0$ and $n$ large enough. 
Since, by the  event $\mathcal{B}$, each pair $(U_{t-1},U_t)$ is not a $\kappa$-marked edge, we then have by
our choice of $\kappa=\epsilon/{12}$,
\begin{align} \label{eq:at-most-kappa}
|I_{t-1} \bigtriangleup I_t| = \|I_t-I_{t-1}\|_2^2 &\le \kappa \E[\|C\circ \Phi(U)\|_2^2] \le 3 \kappa \phi_{n,d}
\le (\epsilon/4)\phi_{n,d}.
\end{align}

Now we construct a collection of independent sets $J_1,\ldots,J_T$ violating the event $\mathcal{A}$,
thus obtaining a contradiction.
To this end,  let $J_1=I_0$. By \eqref{eq:intersection-union} applied for $R=\{0\}$ which holds since the event
 $\mathcal{A}$ holds, we have $|I_M\cap J_1|\le (\epsilon/4)\phi_{n,d}$.
Let $0\le s\le M$ be the smallest index for which $|I_s\cap J_1|\le (\epsilon/4)\phi_{n,d}$, which
exists by the above and the fact $|I_0\cap J_1|=|I_0|=|J_1|\ge (1+\epsilon)\phi_{n,d}$. 
Clearly $s\ge 1$. Then $|I_{s-1}\cap J_1|> (\epsilon/4)\phi_{n,d}$.
Combining this with (\ref{eq:at-most-kappa}) applied to $t=s$, we have that
\begin{align*}
|I_{s-1}\cap J_1|&=|\left((I_{s-1}\setminus I_{s})\cap J_1\right)\cup \left((I_{s-1}\cap I_{s})\cap J_1\right)| \\
&\le |I_{s-1}\setminus I_{s}|+|I_{s}\cap J_1| \le (\epsilon/2)\phi_{n,d}.
\end{align*}
In conclusion,
\begin{align}\label{eq:in-gap}
|I_{s-1}\cap J_1|\in [(\epsilon/4)\phi_{n,d},(\epsilon/2)\phi_{n,d}].
\end{align}
Let $J_2=I_{s-1}$. 

Next we construct $J_3$. By (\ref{eq:intersection-union}) applied for $R=\{0,s-1\}$ we have
\begin{align*}
|I_{s-1+M}\cap (J_1\cup J_2)|\le (\epsilon/4)\phi_{n,d}.
\end{align*}
Repeating the earlier argument, we can find $s \le u \le s-1+M \le 2M$ such that (\ref{eq:in-gap}) 
holds with $u$ replacing $s$ and $J_1\cup J_2$ replacing $J_1$, and we set $J_3=J_{u-1}$. 
We continue along similar lines until the sequence $J_1,\ldots,J_T$ is constructed.

We have constructed a sequence of $T$ time indices in $0, \ldots, TM$ and  independent
sets $J_1,\ldots,J_T$ in graphs corresponding to these time indices 
each with size at least $(1+\epsilon)\phi_{n,d}$ satisfying
\begin{align*}
|J_t\setminus (\cup_{1\le \ell<t} J_\ell)|\in \left[{\epsilon \over 4}\phi_{n,d},{\epsilon \over 2}\phi_{n,d}\right]
\end{align*}
for $t \ge 2$. This violates the m-e-OGP and therefore the event $\mathcal{A}$. We conclude that $\mathcal{A}\cap\mathcal{B}=\emptyset$
as claimed. 
\end{proof}

\begin{proof}[Proof of Theorem~\ref{theorem:no-bad}]
Let $C_j\circ\Phi(u)$ be the $j$-th output of the vector $C\circ \Phi(u)$ with $j\in [n]$. 
Note that each $C_j\circ\Phi$ is itself a circuit of depth at most $p(n)+3$ and some size $s(C_j\circ\Phi)$. 
Thus by the influence version of the Linal--Mansour--Nisan theorem  \cite[Thm 4.30]{o-book}, 
there is a universal constant $\bar c>0$ 
 such that the influence of each $C_j$  satisfies
\[ 
I(C_j\circ\Phi) \leq ( \bar c \log s(C_j\circ \Phi))^{p(n)+2}\leq (\bar c\log s(C\circ\Phi))^{p(n)+2}.
\]
Summing this in $n$ and using the fact that without loss of generality $s(C\circ \Phi)\geq n$, we see that applying Theorem~\ref{thm:no-jumps-vs-influence} with $g_j=C_j, j\in [n]$ yields
\begin{align*}
\pr\left(\cap_{t\in [TM]}\mathcal{E}_t^c\right)\ge 
\exp\left(-{4(\log 2)T n\over\kappa \E\left[\|C\circ \Phi(U)\|_2^2\right]} 
\left(\log s(C\circ \Phi)\right)^{(1+\eta)p(n)}\right),
\end{align*}
for any $\eta>0$ and $n$ large enough.

Now, $\|C\circ \Phi(U)\|_2^2$ is the size of the independent set produced by the 
circuit with graph $\G(n,d/n)$ as an input (or more precisely with $U$ as an input) which, by assumption, is  a
constant factor of $\rho$ away from optimality. So by Frieze's bound \eqref{eq:Max-IS-ER} its expectation is at least $\rho \gamma n$ for some constant $\gamma$ which depends
on $d$.  Thus we have that for $n$ large enough,
\begin{align*}
\pr\left(\cap_{t\in [TM]}\mathcal{E}_t^c\right)\ge \exp\left(-{c T \over  \rho \kappa } 
\left(\log s(C\circ \Phi)\right)^{(1+\eta)p(n)}\right),
\end{align*}
for some $c$ which depends only on $d$.
\end{proof}

\appendix

\section{More on Low-Degree Algorithms}
\label{app:low-deg-alg}

Here we discuss some classes of algorithms that can be represented by, or approximated by, low-degree polynomials. We provide proof sketches of how to write these algorithms as polynomials and discuss what degree is required to do so. We consider polynomials $f: \RR^m \to \RR^n$ of the form $f(Y) = (f_1(Y),\ldots,f_n(Y))$ where each $f_j$ is a polynomial of degree (at most) $D = D(n)$. Here $m = m(n)$ is the dimension of the input. We allow $f$ to have random coefficients (although they must be independent from the input $Y$).

\paragraph{Spectral methods.}

Consider the following general class of spectral methods. For some $N = N(n)$, let $M = M(Y)$ be an $N \times N$ matrix whose entries are polynomials in $Y$ of degree at most $d = O(1)$. Then compute the leading eigenvector of $M$. The leading eigenvector can be approximated via the power iteration scheme: iterate $u^{t+1} \leftarrow M u^t$ starting from some initial vector $u^0$ (which may be random but does not depend on $Y$). After $t$ rounds of power iteration, the result $M^t u^0$ is a (random) polynomial in $Y$ of degree at most $td$.

One particularly simple random optimization problems is to maximize the quadratic form of a random matrix:
\begin{equation}\label{eq:max-eig}
\max_{\|x\|=1} x^\top Y x
\end{equation}
where $Y$ is a GOE matrix, i.e., symmetric $n \times n$ with $Y_{ij} = Y_{ji} \sim \mathcal{N}(0,1/n)$ and $Y_{ii} \sim \mathcal{N}(0,2/n)$. This is equivalent to the $p=2$ case of the spherical $p$-spin model. In the limit $n \to \infty$, the eigenvalues of $Y$ follow the semicircle law on $[-2,2]$, and in particular the optimal value of~\eqref{eq:max-eig} converges to $2$. In order to find an $x \in \RR^n$ achieving $x^\top Y x / \|x\|_2^2 \ge 2 - \varepsilon$ for a constant $\varepsilon > 0$, it suffices to run power iteration for a constant $C = C(\varepsilon)$ number of rounds: $x = Y^C u^0$ where $u^0 \sim \mathcal{N}(0,I)$. While this $x$ does not satisfy the constraint $\|x\|_2 = 1$ exactly, $\|x\|_2$ concentrates tightly around its expectation and so by rescaling we can arrange to have $\|x\|_2 = 1 \pm o(1)$ with high probability. (As we have done throughout this paper, we only ask that a low-degree polynomial outputs an \emph{approximately} valid solution in the appropriate problem-specific sense. This solution can then be ``rounded'' to a valid solution via some canonical procedure, which in this case is simply normalization: $x \leftarrow x/\|x\|_2$.) Thus, a near-optimal solution to~\eqref{eq:max-eig} can be obtained via a constant-degree polynomial.

\paragraph{Iterative methods and AMP.}

Suppose we start with a random initialization $u^0$ and carry out the iteration scheme $u^{t+1} \leftarrow F_t(u^0,\ldots,u^t;Y)$ where $F_t$ is a degree-$d$ polynomial in all of its inputs. Then $u^t$ can be written as a (random) polynomial in $Y$ of degree at most $d^t$. In contrast to the linear update step in power iteration, here the degree can grow exponentially in the number of iterations.

Approximate message passing (AMP) \cite{amp-cs} is a class of iterative methods that give state-of-the-art performance for a variety of statistical tasks, both for estimation problems with a planted signal and for (un-planted) random optimization problems. AMP iterations typically involve certain non-linear transformations, applied entrywise to the current state vector. By replacing each non-linear function with a polynomial of large constant degree, one can approximate each AMP iteration arbitrarily well by a constant-degree polynomial. The existing AMP algorithms for spin glass optimization problems~\cite{montanari-sk,EMS-opt} only need to run for a constant $C(\varepsilon)$ number of iterations in order to reach objective value $\mu^* - \varepsilon$ (for some problem-specific threshold $\mu^*$). As in the previous discussion, AMP is guaranteed to output a nearly-valid solution in the appropriate sense. Thus, the AMP algorithms for spin glass optimization are captured by constant-degree polynomials: the objective value $\mu^* - \varepsilon$ can be obtained by a polynomial of degree $D(\varepsilon)$.

\paragraph{Local algorithms on sparse graphs.}

Now consider the setting where the input is a random graph of constant average degree, e.g., a random $d$-regular graph or an \ER $\G(n,d/n)$ graph, where $d$ is a constant. The algorithm's output is a vector in $\RR^n$, which could be, for example, (an approximation of) the indicator vector for an independent set. We consider ``local algorithms'' in the sense of the \emph{factors of i.i.d.}\ model~\cite{LauerWormald}, defined as follows. For the algorithm's internal use, we attach a random variable $z_i$ to each vertex $i$; these are i.i.d.\ from some distribution of the algorithm designer's choosing. For a constant $r$, an algorithm is called $r$-\emph{local} if the output associated to each vertex depends only on the radius-$r$ neighborhood of that vertex in the graph, including both the graph topology and the labels $z_i$.

We can see that any $O(1)$-local algorithm can be approximated by a constant-degree random polynomial as follows. Suppose we have an $r$-local algorithm $L(Y)$, where $Y$ is the adjacency matrix of the graph. Let $\Delta$ be a constant. For each $i$ there is a random polynomial $f_i(Y)$ of constant degree $D(r,\Delta)$ such that $f_i(Y) = L_i(Y)$ whenever the radius-$r$ neighborhood of vertex $i$ contains at most $\Delta$ vertices. We can construct $f_i$ as follows. It is sufficient to consider fixed $\{z_i\}$, since these can be absorbed into the randomness of $f_i$. Thus, $L_i(Y)$ is determined by the radius-$r$ neighborhood of $i$, i.e., the subgraph consisting of edges reached within distance $r$ from $i$. Consider a fixed graph $\mathcal{N}$ (on the same vertex set as $Y$) spanning at most $\Delta$ vertices. The $\{0,1\}$-indicator that $\mathcal{N}$ is a subgraph of the radius-$r$ neighborhood of $i$ can be written as a constant-degree polynomial (where the degree is the number of edges in $\mathcal{N}$). We can now form $f_i$ by summing over all possibilities for $\mathcal{N}$ and using inclusion-exclusion. By choosing $\Delta$ large (compared to $d$), we can ensure that $f$ agrees with $L$ except on an arbitrarily small constant fraction of vertices. (Note that a small fraction of errors of this type are allowed by our lower bounds.) Thus, local algorithms of constant radius can be captured by constant-degree polynomials. The above proof sketch has been made formal in~\cite{wein2020optimal}.

\paragraph{Exhaustive search.}

One (computationally-inefficient) way to solve an optimization problem is by exhaustive search over all possible solutions. It will be instructive to investigate at what degree such an algorithm can be approximated by a polynomial. As an example, consider the Ising $p$-spin optimization problem
\begin{equation}\label{eq:max-ising}
\max_{x \in \Sigma_n} \langle Y,x^{\otimes p} \rangle
\end{equation}
where $\Sigma_n = \{+1,-1\}^n$ and $Y$ is a $p$-tensor with i.i.d.\ $\mathcal{N}(0,1)$ entries. In time $\exp(O(n))$, one can enumerate all possible solutions and solve~\eqref{eq:max-ising} exactly. To approximate this by a polynomial, consider
\[ f(Y) = \sum_{x \in \Sigma_n} x \langle Y,x^{\otimes p} \rangle^{2k} \]
for some power $k = k(n) \in \NN$. Note that if $k$ is large enough, the sum will be dominated by the term corresponding to the $x^*$ which maximizes~\eqref{eq:max-ising}, and so $f(Y)$ will be approximately equal to a multiple of $x^*$. For the spherical $p$-spin optimization problem, one can similarly replace the sum over $\Sigma_n$ by an integral over the sphere $\cS_n$. A heuristic calculation based on the known behavior of near-maximal states~\cite{extremal-pspin} indicates that $k$ should be chosen as $k = n^{1+o(1)}$ in order for $f(Y)$ to be close to the optimizer, in which case $f$ has degree $n^{1+o(1)}$ (we thank Eliran Subag and Jean-Christophe Mourrat for helpful discussions surrounding this point). This is consistent with a phenomenon that has been observed in various hypothesis testing settings (see~\cite{sam-thesis,lowdeg-notes,subexp-sparse}): the class of degree-$n^\delta$ polynomials is at least as powerful as all known $\exp(n^{\delta-o(1)})$-time algorithms.

\paragraph{Comparison to the hypothesis testing setting.}

Above, we have discussed how algorithms for random optimization problems can be represented as polynomials. The original motivation for this comes from the well-established theory of low-degree algorithms for hypothesis testing problems (see~\cite{sam-thesis,lowdeg-notes}). In the hypothesis testing setting, it is typical for state-of-the-art polynomial-time algorithms to take the form of spectral methods, i.e., the algorithm decides whether the input $Y$ was drawn from the null or planted distribution by thresholding the leading eigenvalue of some matrix $M = M(Y)$ whose entries are constant-degree polynomials in $Y$. To approximate this by a polynomial, consider $f(Y) = \Tr(M^{2k}) = \sum_i \lambda_i^{2k}$ for some power $k = k(n) \in \NN$, where $\{\lambda_i\}$ denote the eigenvalues of $M$. If $k$ is large enough, $f(Y)$ is dominated by the leading eigenvalue. Typically, under the planted distribution there is an outlier eigenvalue that exceeds, by a constant factor, the largest (in magnitude) eigenvalue that occurs under the null distribution. As a result, $k$ should be chosen as $k = \Theta(\log n)$ in order for $f$ to consistently distinguish the planted and null distributions in the appropriate formal sense (see~\cite[Theorem~4.4]{lowdeg-notes} for details). For this reason, low-degree algorithms in the hypothesis testing setting typically require degree $O(\log n)$. In contrast, we have seen above that for random optimization problems with no planted signal, constant degree seems to often be sufficient.

\section*{Acknowledgments}

For helpful discussions, A.S.W.\ is grateful to L\'eo Miolane, Eliran Subag, Jean-Christophe Mourrat, and Ahmed El Alaoui.
The authors are very grateful to Benjamin Rossman for educating us on the state of the art bounds in circuit complexity.
We are also very grateful to Alexander Razborov for reading an earlier version of the paper and providing several comments regarding the presentation of our results on Boolean circuit.

\bibliographystyle{alpha}
\bibliography{main,bibliography}

\newcommand{\etalchar}[1]{$^{#1}$}
\begin{thebibliography}{DKWB19}

\bibitem[AB09]{arora2009computational}
Sanjeev Arora and Boaz Barak.
\newblock {\em Computational complexity: a modern approach}.
\newblock Cambridge University Press, 2009.

\bibitem[AB{\v{C}}13]{ABC13}
Antonio Auffinger, G\'{e}rard {Ben Arous}, and Ji\v{r}\'{\i} {\v{C}}ern\'{y}.
\newblock Random matrices and complexity of spin glasses.
\newblock {\em Comm. Pure Appl. Math.}, 66(2):165--201, 2013.

\bibitem[AC18]{AuffChen18}
Antonio Auffinger and Wei-Kuo Chen.
\newblock On the energy landscape of spherical spin glasses.
\newblock {\em Advances in Mathematics}, 330:553--588, 2018.

\bibitem[ART06]{achlioptas2006solution}
Dimitris Achlioptas and Federico Ricci-Tersenghi.
\newblock On the solution-space geometry of random constraint satisfaction
  problems.
\newblock In {\em Proceedings of the thirty-eighth annual ACM symposium on
  Theory of computing}, pages 130--139, 2006.

\bibitem[AS04]{alon2004probabilistic}
Noga Alon and Joel~H Spencer.
\newblock {\em The probabilistic method}.
\newblock John Wiley \& Sons, 2004.

\bibitem[BCKM98]{BCKM98}
Jean-Philippe Bouchaud, Leticia~F Cugliandolo, Jorge Kurchan, and Marc
  M{\'e}zard.
\newblock Out of equilibrium dynamics in spin-glasses and other glassy systems.
\newblock {\em Spin glasses and random fields}, pages 161--223, 1998.

\bibitem[BCR20]{replicated-gradient}
Giulio Biroli, Chiara Cammarota, and Federico {Ricci-Tersenghi}.
\newblock How to iron out rough landscapes and get optimal performances:
  Averaged gradient descent and its application to tensor {PCA}.
\newblock {\em Journal of Physics A: Mathematical and Theoretical}, 2020.

\bibitem[BDG01]{BADG01}
G{\'e}rard {Ben Arous}, Amir Dembo, and Alice Guionnet.
\newblock Aging of spherical spin glasses.
\newblock {\em Probab. Theory Related Fields}, 120(1):1--67, 2001.

\bibitem[BDG06]{BADG06}
G{\'e}rard {Ben Arous}, Amir Dembo, and Alice Guionnet.
\newblock Cugliandolo-{K}urchan equations for dynamics of spin-glasses.
\newblock {\em Probab. Theory Related Fields}, 136(4):619--660, 2006.

\bibitem[{Ben}02]{BA02}
G{\'e}rard {Ben Arous}.
\newblock Aging and spin-glass dynamics.
\newblock In {\em Proceedings of the {I}nternational {C}ongress of
  {M}athematicians, {V}ol. {III} ({B}eijing, 2002)}, pages 3--14. Higher Ed.
  Press, Beijing, 2002.

\bibitem[BGJ20a]{algorithmic-tensor}
G\'erard {Ben Arous}, Reza Gheissari, and Aukosh Jagannath.
\newblock Algorithmic thresholds for tensor pca.
\newblock {\em The Annals of Probability}, 48(4):2052--2087, 2020.

\bibitem[BGJ20b]{BGJ20}
G{\'e}rard {Ben Arous}, Reza Gheissari, and Aukosh Jagannath.
\newblock Bounding flows for spherical spin glass dynamics.
\newblock {\em Communications in Mathematical Physics}, 373(3):1011--1048,
  2020.

\bibitem[BGT13]{BayatiGamarnikTetali}
Mohsen Bayati, David Gamarnik, and Prasad Tetali.
\newblock Combinatorial approach to the interpolation method and scaling limits
  in sparse random graphs.
\newblock {\em Annals of Probability. (Conference version in Proc. 42nd Ann.
  Symposium on the Theory of Computing (STOC) 2010)}, 41:4080--4115, 2013.

\bibitem[BH21]{bresler2021algorithmic}
Guy Bresler and Brice Huang.
\newblock The algorithmic phase transition of random {k-SAT} for low degree
  polynomials.
\newblock {\em arXiv preprint arXiv:2106.02129}, 2021.

\bibitem[BHK{\etalchar{+}}19]{p-cal}
Boaz Barak, Samuel Hopkins, Jonathan Kelner, Pravesh~K Kothari, Ankur Moitra,
  and Aaron Potechin.
\newblock A nearly tight sum-of-squares lower bound for the planted clique
  problem.
\newblock {\em SIAM Journal on Computing}, 48(2):687--735, 2019.

\bibitem[BJ18]{arous2018spectral}
G\'{e}rard {Ben Arous} and Aukosh Jagannath.
\newblock Spectral gap estimates in mean field spin glasses.
\newblock {\em Communications in Mathematical Physics}, 361(1):1--52, 2018.

\bibitem[BKW19]{sk-cert}
Afonso~S Bandeira, Dmitriy Kunisky, and Alexander~S Wein.
\newblock Computational hardness of certifying bounds on constrained {PCA}
  problems.
\newblock {\em arXiv preprint arXiv:1902.07324}, 2019.

\bibitem[CGPR19]{chen2019suboptimality}
Wei-Kuo Chen, David Gamarnik, Dmitry Panchenko, and Mustazee Rahman.
\newblock Suboptimality of local algorithms for a class of max-cut problems.
\newblock {\em The Annals of Probability}, 47(3):1587--1618, 2019.

\bibitem[CHH17]{coja2017walksat}
Amin {Coja-Oghlan}, Amir Haqshenas, and Samuel Hetterich.
\newblock Walksat stalls well below satisfiability.
\newblock {\em SIAM Journal on Discrete Mathematics}, 31(2):1160--1173, 2017.

\bibitem[CHS93]{crisanti1993sphericalp}
A~Crisanti, H~Horner, and H-J Sommers.
\newblock The spherical p-spin interaction spin-glass model.
\newblock {\em Zeitschrift f{\"u}r Physik B Condensed Matter}, 92(2):257--271,
  1993.

\bibitem[CK93]{CugKur93}
Leticia~F. Cugliandolo and Jorge Kurchan.
\newblock Analytical solution of the off-equilibrium dynamics of a long-range
  spin-glass model.
\newblock {\em Phys. Rev. Lett.}, 71:173--176, Jul 1993.

\bibitem[COE11]{coja2011independent}
A.~Coja-Oghlan and C.~Efthymiou.
\newblock On independent sets in random graphs.
\newblock In {\em Proceedings of the Twenty-Second Annual ACM-SIAM Symposium on
  Discrete Algorithms}, pages 136--144. SIAM, 2011.

\bibitem[CS17]{ChenSen17}
Wei-Kuo Chen and Arnab Sen.
\newblock Parisi formula, disorder chaos and fluctuation for the ground state
  energy in the spherical mixed p-spin models.
\newblock {\em Communications in Mathematical Physics}, 350(1):129--173, 2017.

\bibitem[Cug03]{Cug03}
Leticia~F Cugliandolo.
\newblock Course 7: Dynamics of glassy systems.
\newblock In {\em Slow Relaxations and nonequilibrium dynamics in condensed
  matter}, pages 367--521. Springer, 2003.

\bibitem[DKWB19]{subexp-sparse}
Yunzi Ding, Dmitriy Kunisky, Alexander~S Wein, and Afonso~S Bandeira.
\newblock Subexponential-time algorithms for sparse {PCA}.
\newblock {\em arXiv preprint arXiv:1907.11635}, 2019.

\bibitem[DMM09]{amp-cs}
David~L Donoho, Arian Maleki, and Andrea Montanari.
\newblock Message-passing algorithms for compressed sensing.
\newblock {\em Proceedings of the National Academy of Sciences},
  106(45):18914--18919, 2009.

\bibitem[EAMS21]{EMS-opt}
Ahmed El~Alaoui, Andrea Montanari, and Mark Sellke.
\newblock Optimization of mean-field spin glasses.
\newblock {\em The Annals of Probability}, 49(6):2922--2960, 2021.

\bibitem[FJS91]{furst1991improved}
Merrick~L Furst, Jeffrey~C Jackson, and Sean~W Smith.
\newblock Improved learning of {AC0} functions.
\newblock In {\em COLT}, volume~91, pages 317--325, 1991.

\bibitem[Fri90]{FriezeIndependentSet}
Alan Frieze.
\newblock On the independence number of random graphs.
\newblock {\em Discrete Mathematics}, 81:171--175, 1990.

\bibitem[Gam21]{gamarnik2021overlap}
David Gamarnik.
\newblock The overlap gap property: A topological barrier to optimizing over
  random structures.
\newblock {\em Proceedings of the National Academy of Sciences}, 118(41), 2021.

\bibitem[GJ19]{gheissari2019spectral}
Reza Gheissari and Aukosh Jagannath.
\newblock On the spectral gap of spherical spin glass dynamics.
\newblock In {\em Annales de l'Institut Henri Poincar{\'e}, Probabilit{\'e}s et
  Statistiques}, volume~55, pages 756--776. Institut Henri Poincar{\'e}, 2019.

\bibitem[GJ21]{gamarnik2019overlap}
David Gamarnik and Aukosh Jagannath.
\newblock The overlap gap property and approximate message passing algorithms
  for $ p $-spin models.
\newblock {\em The Annals of Probability}, 49(1):180--205, 2021.

\bibitem[GJW20]{gamarnik2020lowFOCS}
David Gamarnik, Aukosh Jagannath, and Alexander~S Wein.
\newblock Low-degree hardness of random optimization problems.
\newblock In {\em 61st Annual Symposium on Foundations of Computer Science},
  2020.

\bibitem[GJW21]{gamarnik2021circuit}
David Gamarnik, Aukosh Jagannath, and Alexander~S Wein.
\newblock Circuit lower bounds for the p-spin optimization problem.
\newblock {\em arXiv preprint arXiv:2109.01342}, 2021.

\bibitem[GS17a]{gamarnik2014limits}
David Gamarnik and Madhu Sudan.
\newblock Limits of local algorithms over sparse random graphs.
\newblock {\em Annals of Probability}, 45:2353--2376, 2017.

\bibitem[GS17b]{gamarnik2017performance}
David Gamarnik and Madhu Sudan.
\newblock Performance of sequential local algorithms for the random {NAE-K-SAT}
  problem.
\newblock {\em SIAM Journal on Computing}, 46(2):590--619, 2017.

\bibitem[GT02]{Ton02}
Francesco Guerra and Fabio~Lucio Toninelli.
\newblock The thermodynamic limit in mean field spin glass models.
\newblock {\em Comm. Math. Phys.}, 230(1):71--79, 2002.

\bibitem[Gui07]{Gui07}
Alice Guionnet.
\newblock Dynamics for spherical models of spin-glass and aging.
\newblock In {\em Spin glasses}, volume 1900 of {\em Lecture Notes in Math.},
  pages 117--144. Springer, Berlin, 2007.

\bibitem[Has86]{hastad1986almost}
Johan Hastad.
\newblock Almost optimal lower bounds for small depth circuits.
\newblock In {\em Proceedings of the eighteenth annual ACM symposium on Theory
  of computing}, pages 6--20, 1986.

\bibitem[Has20]{hastings-quantum}
Matthew~B Hastings.
\newblock Classical and quantum algorithms for tensor principal component
  analysis.
\newblock {\em Quantum}, 4:237, 2020.

\bibitem[HKP{\etalchar{+}}17]{sos-hidden}
Samuel~B Hopkins, Pravesh~K Kothari, Aaron Potechin, Prasad Raghavendra, Tselil
  Schramm, and David Steurer.
\newblock The power of sum-of-squares for detecting hidden structures.
\newblock In {\em 2017 IEEE 58th Annual Symposium on Foundations of Computer
  Science (FOCS)}, pages 720--731. IEEE, 2017.

\bibitem[Hop18]{sam-thesis}
Samuel Hopkins.
\newblock {\em Statistical Inference and the Sum of Squares Method}.
\newblock PhD thesis, Cornell University, 2018.

\bibitem[HS17]{HS-bayesian}
Samuel~B Hopkins and David Steurer.
\newblock Efficient bayesian estimation from few samples: community detection
  and related problems.
\newblock In {\em 2017 IEEE 58th Annual Symposium on Foundations of Computer
  Science (FOCS)}, pages 379--390. IEEE, 2017.

\bibitem[HS21]{huang2021tight}
Brice Huang and Mark Sellke.
\newblock Tight lipschitz hardness for optimizing mean field spin glasses.
\newblock {\em arXiv preprint arXiv:2110.07847}, 2021.

\bibitem[HSS15]{tensor-pca-sos}
Samuel~B Hopkins, Jonathan Shi, and David Steurer.
\newblock Tensor principal component analysis via sum-of-square proofs.
\newblock In {\em Conference on Learning Theory}, pages 956--1006, 2015.

\bibitem[HSSS16]{fast-sos}
Samuel~B Hopkins, Tselil Schramm, Jonathan Shi, and David Steurer.
\newblock Fast spectral algorithms from sum-of-squares proofs: tensor
  decomposition and planted sparse vectors.
\newblock In {\em Proceedings of the forty-eighth annual ACM symposium on
  Theory of Computing}, pages 178--191, 2016.

\bibitem[Hsu02]{Hsu02}
Elton~P. Hsu.
\newblock {\em Stochastic analysis on manifolds}, volume~38 of {\em Graduate
  Studies in Mathematics}.
\newblock American Mathematical Society, Providence, RI, 2002.

\bibitem[Jag19]{jagannath2019dynamics}
Aukosh Jagannath.
\newblock Dynamics of mean field spin glasses on short and long timescales.
\newblock {\em Journal of Mathematical Physics}, 60(8):083305, 2019.

\bibitem[Jan97]{janson-gaussian}
Svante Janson.
\newblock {\em Gaussian hilbert spaces}, volume 129.
\newblock Cambridge university press, 1997.

\bibitem[JS20]{JS17}
Aukosh Jagannath and Subhabrata Sen.
\newblock On the unbalanced cut problem and the generalized
  sherrington--kirkpatrick model.
\newblock {\em Annales de l'Institut Henri Poincar{\'e} D}, 8(1):35--88, 2020.

\bibitem[JT17]{JagTob17}
Aukosh Jagannath and Ian Tobasco.
\newblock Low temperature asymptotics of spherical mean field spin glasses.
\newblock {\em Comm. Math. Phys.}, 352(3):979--1017, 2017.

\bibitem[Kar76]{karp}
Richard~M Karp.
\newblock The probabilistic analysis of some combinatorial search algorithms.
\newblock 1976.

\bibitem[KWB19]{lowdeg-notes}
Dmitriy Kunisky, Alexander~S Wein, and Afonso~S Bandeira.
\newblock Notes on computational hardness of hypothesis testing: Predictions
  using the low-degree likelihood ratio.
\newblock {\em arXiv preprint arXiv:1907.11636}, 2019.

\bibitem[LMN93]{LMN}
Nathan Linial, Yishay Mansour, and Noam Nisan.
\newblock Constant depth circuits, fourier transform, and learnability.
\newblock {\em Journal of the ACM (JACM)}, 40(3):607--620, 1993.

\bibitem[LRR17]{li2017ac}
Yuan Li, Alexander Razborov, and Benjamin Rossman.
\newblock On the {AC0} complexity of subgraph isomorphism.
\newblock {\em SIAM Journal on Computing}, 46(3):936--971, 2017.

\bibitem[LT11]{LedouxTalagrand}
Michel Ledoux and Michel Talagrand.
\newblock {\em Probability in {B}anach spaces}.
\newblock Classics in Mathematics. Springer-Verlag, Berlin, 2011.
\newblock Isoperimetry and processes, Reprint of the 1991 edition.

\bibitem[LW07]{LauerWormald}
J.~Lauer and N.C. Wormald.
\newblock Large independent sets in regular graphs of large girth.
\newblock {\em Journal of Combinatorial Theory, Series B}, 97:999--1009, 2007.

\bibitem[MMZ05]{mezard2005clustering}
M.~M{\'e}zard, T.~Mora, and R.~Zecchina.
\newblock Clustering of solutions in the random satisfiability problem.
\newblock {\em Physical Review Letters}, 94(19):197205, 2005.

\bibitem[Mon19]{montanari-sk}
Andrea Montanari.
\newblock Optimization of the {Sherrington-Kirkpatrick} hamiltonian.
\newblock In {\em 2019 IEEE 60th Annual Symposium on Foundations of Computer
  Science (FOCS)}, pages 1417--1433. IEEE, 2019.

\bibitem[MPS{\etalchar{+}}84]{Mez84}
Marc M\'ezard, Giorgio Parisi, Nicolas Sourlas, G\'erard Toulouse, and Miguel
  Virasoro.
\newblock {Replica Symmetry-Breaking and the Nature of the Spin-Glass Phase}.
\newblock {\em J. Physique}, 45:843, 1984.

\bibitem[O'D14]{o-book}
Ryan O'Donnell.
\newblock {\em Analysis of boolean functions}.
\newblock Cambridge University Press, 2014.

\bibitem[Pan13]{panchenko2013sherrington}
Dmitry Panchenko.
\newblock {\em The Sherrington-Kirkpatrick model}.
\newblock Springer Science \& Business Media, 2013.

\bibitem[RM14]{RM-tensor}
Emile Richard and Andrea Montanari.
\newblock A statistical model for tensor {PCA}.
\newblock In {\em Advances in Neural Information Processing Systems}, pages
  2897--2905, 2014.

\bibitem[Ros10]{rossman2010average}
Benjamin Rossman.
\newblock {\em Average-case complexity of detecting cliques}.
\newblock PhD thesis, Massachusetts Institute of Technology, 2010.

\bibitem[Ros18]{rossman2018lower}
Benjamin Rossman.
\newblock Lower bounds for subgraph isomorphism.
\newblock In {\em Proceedings of the International Congress of Mathematicians:
  Rio de Janeiro 2018}, pages 3425--3446. World Scientific, 2018.

\bibitem[RSS18]{sos-survey}
Prasad Raghavendra, Tselil Schramm, and David Steurer.
\newblock High dimensional estimation via sum-of-squares proofs.
\newblock In {\em Proceedings of the International Congress of Mathematicians:
  Rio de Janeiro 2018}, pages 3389--3423. World Scientific, 2018.

\bibitem[RV17]{rahman2017local}
Mustazee Rahman and Balint Virag.
\newblock Local algorithms for independent sets are half-optimal.
\newblock {\em The Annals of Probability}, 45(3):1543--1577, 2017.

\bibitem[Sel21]{sellke2021optimizing}
Mark Sellke.
\newblock Optimizing mean field spin glasses with external field.
\newblock {\em arXiv preprint arXiv:2105.03506}, 2021.

\bibitem[Sip97]{sipser}
M.~Sipser.
\newblock {\em Introduction to the Theory of Computability}.
\newblock PWS Publishing Company, 1997.

\bibitem[Sub21]{subag-full-rsb}
Eliran Subag.
\newblock Following the ground states of full-rsb spherical spin glasses.
\newblock {\em Communications on Pure and Applied Mathematics},
  74(5):1021--1044, 2021.

\bibitem[SZ17]{extremal-pspin}
Eliran Subag and Ofer Zeitouni.
\newblock The extremal process of critical points of the pure p-spin spherical
  spin glass model.
\newblock {\em Probability theory and related fields}, 168(3-4):773--820, 2017.

\bibitem[Tal17]{tal2017tight}
Avishay Tal.
\newblock Tight bounds on the fourier spectrum of {AC0}.
\newblock In {\em 32nd Computational Complexity Conference (CCC 2017)}. Schloss
  Dagstuhl-Leibniz-Zentrum fuer Informatik, 2017.

\bibitem[Tes12]{Tes12}
Gerald Teschl.
\newblock {\em Ordinary differential equations and dynamical systems}, volume
  140 of {\em Graduate Studies in Mathematics}.
\newblock American Mathematical Society, Providence, RI, 2012.

\bibitem[Var16]{varadhan}
S.~R.~S. Varadhan.
\newblock {\em Large deviations}, volume~27 of {\em Courant Lecture Notes in
  Mathematics}.
\newblock Courant Institute of Mathematical Sciences, New York; American
  Mathematical Society, Providence, RI, 2016.

\bibitem[Ver18]{vershynin2018high}
Roman Vershynin.
\newblock {\em High-dimensional probability: An introduction with applications
  in data science}, volume~47.
\newblock Cambridge University Press, 2018.

\bibitem[Wei22]{wein2020optimal}
Alexander~S Wein.
\newblock Optimal low-degree hardness of maximum independent set.
\newblock {\em Mathematical Statistics and Learning}, 2022.

\bibitem[WEM19]{kikuchi}
Alexander~S Wein, Ahmed {El Alaoui}, and Cristopher Moore.
\newblock The {Kikuchi} hierarchy and tensor {PCA}.
\newblock In {\em 2019 IEEE 60th Annual Symposium on Foundations of Computer
  Science (FOCS)}, pages 1446--1468. IEEE, 2019.

\end{thebibliography}

\end{document}